\documentclass[12pt]{article}
\usepackage[utf8]{inputenc}
\usepackage[margin=1in]{geometry}

\usepackage{graphicx}
\usepackage{natbib}

\usepackage{amsmath, amsthm, amsfonts, amssymb, latexsym, mathtools}
\usepackage[shortlabels]{enumitem}
\usepackage{bm}
\usepackage{fancyvrb}
\usepackage{longtable}
\usepackage{eurosym}

\usepackage{color}

\usepackage{setspace}
\theoremstyle{plain}
\newtheorem{thm}{Theorem}[section]

\newtheorem{lem}[thm]{Lemma}

\theoremstyle{definition}
\newtheorem{rem}{Remark}

\newcommand{\ud}{\textnormal{d}}
\newcommand{\tr}{^{\top}}
\newcommand{\pr}{^{\prime}}
\newcommand{\zero}{\mathbf{0}}
\newcommand{\R}{\mathbb{R}}

\newcommand{\calN}{\mathcal{N}}

\newcommand{\pconv}{\xrightarrow{p}}

\newcommand{\dconvs}{\xrightarrow{d^\ast}}

\DeclareMathOperator*{\argmin}{argmin}

\newcommand{\prob}[1]{\textnormal{P} \left\{ #1 \right\}}
\newcommand{\probs}[1]{\textnormal{P}^\ast \left\{ #1 \right\}}
\newcommand{\ex}[1]{\textnormal{E} \left[ #1 \right]}
\newcommand{\exs}[1]{\textnormal{E}^\ast \left[ #1 \right]}
\newcommand{\var}[1]{\textnormal{Var} \left( #1 \right)}

\newcommand{\Op}[1]{O_p \left( #1 \right)}
\newcommand{\op}[1]{o_p \left( #1 \right)}
\newcommand{\Ops}[1]{O_{p^\ast} \left( #1 \right)}
\newcommand{\ops}[1]{o_{p^\ast} \left( #1 \right)}

\newcommand{\Var}{\operatorname{Var}}
\newcommand{\RR}{\mathbb{R}}

\newtheorem{theorem}{Theorem}
\newtheorem{lemma}{Lemma}

\newcommand{\inpr}{\overset{p}{\longrightarrow}}
\newcommand{\indist}{\overset{d}{\longrightarrow}}
\newcommand{\MM}{\mathbb{M}}
\newcommand{\VV}{\mathbb{V}}

\newcommand{\blind}{0}

\begin{document}

\if0\blind
{
  \title{\bf Partitioned Wild Bootstrap for Panel Data Quantile Regression\footnote{We are grateful for comments and suggestions from Ivan Canay, Federico Bugni, Patrik Guggenberger, James MacKinnon, and seminar participants at Northwestern University, University of Kentucky, the 2024 Midwest Econometrics Group conference,  SEA 94th Annual Meeting, New York Camp Econometrics XIX, and 30th International Panel Data Conference.}}

  \author{Antonio F. Galvao\footnote{Department of Economics, Michigan State University. agalvao@msu.edu}
    \and Carlos Lamarche\footnote{Department of Economics, University of Kentucky. clamarche@uky.edu}
    \and
    Thomas Parker\footnote{Department of Economics, University of Waterloo. tmparker@uwaterloo.ca}
  }
  \maketitle
} \fi

\if1\blind
{
  \bigskip
  \bigskip
  \bigskip
\bigskip
  \begin{center}
    {\LARGE\bf {\bf Partitioned Wild Bootstrap for Panel Data Quantile Regression}}
\end{center}
  \medskip
} \fi

\begin{abstract}
\begin{singlespace}
\noindent Practical inference procedures for quantile regression models of panel data have been a pervasive concern in empirical work, and can be especially challenging when the panel is observed over many time periods and temporal dependence needs to be taken into account. In this paper, we propose a new bootstrap method that applies random weighting to a partition of the data --- partition-invariant weights are used in the bootstrap data generating process --- to conduct statistical inference for conditional quantiles in panel data that have significant time-series dependence. We demonstrate that the procedure is asymptotically valid for approximating the distribution of the fixed effects quantile regression estimator. The bootstrap procedure offers a viable alternative to existing resampling methods. Simulation studies show numerical evidence that the novel approach has accurate small sample behavior, and an empirical application illustrates its use. 
\end{singlespace}

\vspace{4mm}

\vfill

\end{abstract}

\doublespacing

\newpage

\section{Introduction}

Since the seminal work of \citet{KoenkerBassett78}, quantile regression (QR) models have provided a valuable tool as a way of capturing heterogeneous effects that covariates may exert on the outcome of interest, exposing a wide variety of forms of conditional heterogeneity under weak distributional assumptions.
QR methods have been widely employed to estimate causal effects as well as structural economic models. 

Recently, there has been growing empirical and theoretical literatures on QR models for panel data. \citet{Koenker04} introduced a general approach for estimation of fixed effects quantile regression (FE-QR) models which treats the individual effects as parameters to be estimated. The FE-QR estimator is designed to control for individual-specific heterogeneity while exploring heterogeneous covariate effects, and therefore, it provides a flexible method for practical analysis of panel data. Consistency and weak convergence of the FE-QR estimator have been established in \citet{Kato2012} and \citet{GalvaoGuVolgushev20} under conditions that are similar to the ones considered for nonlinear panel data estimators. Many papers have suggested alternative methods for panel data QR and derived the corresponding statistical properties, including, among others, \citet{Lamarche10}, \citet{Canay11}, \cite{KimYang11}, \citet{GalvaoLamarcheLima13}, \citet{ArellanoBonhomme15}, \citet{GrahamHahnPoirierPowell2015}, \citet{GalvaoKato16}, \citet{ChetverikovLarsenPalmer16}, \citet{GuVolgushev19}, \citet{MachadoSantosSilva19}, and \cite{CHERNOZHUKOV2024}.
A number of these studies rely on large-$T$ panel data, where the number of time periods ($T$) grows faster than the number of cross-sectional units ($N$).

Practical inference procedures for FE-QR have been a pervasive concern in empirical work, especially when the time-series is large and temporal dependence has to be accommodated. Inference methods, which have been developed mostly for independent data, have mainly used asymptotic approximations for the construction of test statistics and the variances of the estimators. Nevertheless, these variances depend on the conditional density of the innovations, which can be difficult to work with in practice and have bandwidth nuisance parameters. Hence, the use of the bootstrap as an alternative to such asymptotic approximations has been considered, but its properties have not received the same amount of attention as in the QR for cross-section or time-series literatures. The sole exceptions are \cite{LamarcheParker23} and \cite{GalvaoParkerXiao24}. \cite{LamarcheParker23} proposed a copula-based approach for a penalized estimator under stationary and $\beta$-mixing conditions, but the implementation of the approach relies on parametric copulas and knowledge of dependence parameters. \cite{GalvaoParkerXiao24} propose a bootstrap based on multiplying one non-negative weight with all of a unit's observations. While this approach has the advantage that serial dependence within each cross-sectional unit is preserved, it can lead to size distortions when practitioners use panels with a small number of cross-sectional units.  The bootstrap literature on the time-series dimension for conditional average linear panel data models has a longer history \citep[e.g.][]{Kapetanios08,Goncalves11,GoncalvesKaffo15}. The analysis of bootstrap algorithms for the linear mean regression model is aided by its linearity and the within transform.  Meanwhile, the model for conditional quantiles with individual effects must be treated as a nonlinear model, and this makes bootstrap inference a challenging problem.

This paper develops a new bootstrap procedure for FE-QR models that is easy to implement. We propose a novel wild residual bootstrap for dependent panel data where one bootstrap weight is applied to each cell in a partition of a unit's observations. Intuitively, reweighting residuals is particularly useful to accommodate conditioning on covariates and fixed effects, while the partition-based approach is aimed at estimating the stationary distribution over a large number of time periods. Hence, we incorporate dependence into bootstrap procedures for panel data FE-QR, improving upon existing results in the literature. This is an important addition, since the statistical properties of panel data QR models are established under large panel approximations.  Our results contribute to the broader QR literature on cross-section and time-series models as in \citet{Hahn95,Hahn97}, \citet{Fitzenberger98}, \citet{Shao10ss}, \citet{FengHeHu11}, \citet{SongRitovHardle12}, \citet{Hagemann17}, and \citet{GregoryLahiriNordman18}. 

We establish the asymptotic validity of the procedure using developments that are different than those used in \cite{LamarcheParker23}, since new results are needed in the dependent case. While the conditions for the marginal weight distribution is standard in the literature \citep{FengHeHu11}, we require a partition cell length condition which is different from the block length conditions considered by \cite{Fitzenberger98}, \cite{Shao10ss} and \cite{GregoryLahiriNordman18}. We show how to empirically choose cell length, which aims to minimize the difference between the bootstrap variance and the asymptotic variance of the FE-QR estimator. Using an extensive simulation study, we find that empirical coverage rates for the proposed bootstrap are close to the nominal counterparts and coverage improves as the sample size increases. 

The remainder of the paper is organized as follows. Section \ref{sec:FEQR} reviews standard inference procedures for FE-QR. In Section \ref{sec:wbb}, we present the new bootstrap procedure and discuss its uses for inference. Section \ref{sec:MC} evaluates the finite sample performance of the bootstrap procedure. In Section \ref{sec:application}, we estimate a panel data quantile regression model and apply the method to evaluate how consumers respond to time-of-use electricity pricing. Finally, Section \ref{sec:conclusion} concludes.

\section{Inference for fixed effects quantile regression}\label{sec:FEQR}

Our model of interest is the fixed effects quantile regression (FE-QR) model. We observe $T$ time periods (indexed by $t$) of jointly stationary data $\{ ( y_{it},\bm{x}_{it}') \}_{t=1}^{T}$ for each of $N$ units (indexed by $i$), where $y_{it} \in \RR$ denotes the response and $\bm{x}_{it}$ denotes a $p$-dimensional vector of covariates for unit $i$ at time $t$. For quantile $\tau \in (0,1)$, the FE-QR model is 
\begin{equation} \label{eq:feqr_model}
  y_{it}  = \bm{x}_{it}' \bm{\beta}_0(\tau) + \alpha_{i0}(\tau) + u_{it},
\end{equation}
where $u_{it}$ is a disturbance whose $\tau$-th quantile conditional on $\bm{x}_{it}$ is equal to zero, implying that the conditional quantile of the response variable is $Q_{y_{it}}(\tau | \bm{x}_{it})  = \bm{x}_{it}' \bm{\beta}_0(\tau) + \alpha_{i0}(\tau)$. The parameter of interest is $\bm{\beta}_0(\tau) \in \mathcal{B} \subseteq \RR^p$, and the scalar individual specific effect, $\alpha_{i0}(\tau)$, is treated as a nuisance parameter. Because we consider just one quantile value, we suppress the dependence of the parameters on $\tau$ in the sequel.

Let $\bm{\theta} = (\bm{\beta}',\bm{\alpha}')' \in \bm{\Theta} \subseteq
\RR^{p+N}$, where $\bm{\alpha} = (\alpha_{1},...,\alpha_{N})'$, and let
$\bm{\theta}_0 = (\bm{\beta}_0',\bm{\alpha}_0')'$. To estimate $\bm{\theta}_0$,
we consider the following FE-QR estimator:
\begin{equation} \label{eq:feqr}
  \hat{\bm{\theta}} = (\hat{\bm{\beta}}', \hat{\bm{\alpha}}')'
  = \argmin_{\bm{\theta} \in \bm{\Theta}} \sum_{i=1}^N \sum_{t=1}^T \rho_{\tau}
  (y_{it} - \bm{x}_{it}' \bm{\beta} - \alpha_i),
\end{equation}
where $\rho _{\tau }(u)=$ $u(\tau -I(u<0))$ is the quantile regression loss
function. The estimator defined in \eqref{eq:feqr} was introduced by \citet{Koenker04}.  Although the parameter of interest is $\bm{\beta}_0$, it is important to note that this estimation strategy estimates all of $\bm{\theta}_0$ and then focuses on $\bm{\beta}_0$ for inference purposes.  This is due to the fact that no known data transformation exists that allows one to avoid estimation of $\bm{\alpha}$ for this model as in, for example, linear conditional average panel models.  It prevents us from naively applying bootstrap algorithms from the time-series literature unit-wise to these observations, as is shown below in Section~\ref{sec:MC}.

The asymptotic distribution of $\hat{\bm{\beta}}$ is described below in Lemma~\ref{lem:feqr}.  It depends on the following regularity conditions.

\begin{enumerate}[label=(A\arabic*), ref=(A\arabic*)]

\item \label{assume:stationary}
  The processes $\{(y_{it},\bm{x}_{it}), t \in 1, 2, \ldots\}$ are strictly stationary and $\beta$-mixing for each $i$ and independent across $i$. Letting $\{\beta_i(k)\}_j$ denote the $\beta$-mixing coefficients, assume that there are constants $0 < a < 1$ and $C > 0$ such that $\sup_i \beta_i(k) \leq C a^k$ for all $k \geq 1$. 

    \item \label{assume:jointdens}
    Let $u_{it} := y_{it} - \bm{x}_{it}' \bm{\beta}_0 - \alpha_{i0}$. The random vector $(u_{it}, u_{it+k})$ has a density conditional on $(\bm{x}_{it}, \bm{x}_{it+k})$ that is bounded uniformly over $i$ and $k = 1, 2, \ldots$
  
    \item \label{a:minf} Let $F_i$ denote the distribution function of $u_{it}$ given $\bm{x}_{it}$, that is, $F_i(u) = \prob{u_{i1} \leq u | \bm{x}_{i1}}$.  The conditional density function $f_i$ corresponding to $F_i$ is uniformly bounded and has a bounded first derivative, that is, $\overline{f} = \sup_i \sup_{u \in \RR, \bm{x} \in \RR^p} | f_i(u | \bm{x}) | < \infty$ and $\overline{f'} = \sup_i \sup_{u \in \RR, \bm{x} \in \RR^p} | f'_i(u | \bm{x}) | < \infty$. Assume that in an open neighborhood $\mathcal{U}$ of $0$, $f_i$ is bounded away from zero for all realizations of $\bm{x}_{it}$: $\underline{f} = \inf_i \inf_{u \in \mathcal{U}, \bm{x} \in \RR^p} | f_i(u | \bm{x}) | < \infty$.

    \item \label{a:boundedX} Assume $\|\bm{x}_{i1}\| \leq M < \infty$ a.s.

    \item \label{assume:Avar_dependent} Assume that $(\alpha_i, \bm{\beta})$ lies in a compact set for all $i$.  Define $\varphi_i = \ex{f_i(0 | \bm{x}_{i1})}$, $\bm{g}_i = \ex{f_i(0 | \bm{x}_{i1}) \bm{x}_{i1}}$ and $\bm{J}_i = \ex{f_i(0 | \bm{x}_{i1}) \bm{x}_{i1} \bm{x}_{i1}'}$. Further define
  \begin{equation*}
    \bm{D}_N =  \frac{1}{N} \sum_{i=1}^N  \left( \bm{J}_i - \varphi_i^{-1} \bm{g}_i \bm{g}_i' \right),\; \;  \bm{V}_{NT} = \frac{1}{N} \sum_{i=1}^N \Var \left( \frac{1}{\sqrt{T}} \sum_{t=1}^T \tilde{\bm{x}}_{it} \psi_{\tau} \left( u_{it} \right) \right),
  \end{equation*}
  where $\tilde{\bm{x}}_{it} = \bm{x}_{it} - \varphi_i^{-1} \bm{g}_i$ and $\psi_\tau(u) = \tau - I( u < 0)$ is the quantile score function. Assume $\bm{D}_N$ is nonsingular for all $N$ and $\bm{D} = \lim_{N \rightarrow \infty} \bm{D}_N$ exists and is nonsingular. Moreover, $\bm{V} = \lim_{N,T \rightarrow \infty} \bm{V}_{NT}$ exists and is nonsingular.
\end{enumerate}

Similar conditions are used in the literature. For instance, a version of Assumption \ref{assume:stationary} has been used in \citet{Kato2012} and \citet{GalvaoGuVolgushev20}. Assumptions \ref{assume:jointdens}--\ref{a:minf} guarantee the convexity of the limiting distribution of the objective function, and therefore, the uniqueness of the fixed effects estimator (see \citet{Kato2012} for slightly more minimal assumptions for consistency). Assumption \ref{a:boundedX} implies appropriate moment conditions on the covariates and it is similar to (B1) in \citet{Kato2012} and (B3) in \cite{LamarcheParker23}. Lastly, \ref{assume:Avar_dependent} assumes the extra regularity conditions needed, beyond those for consistency, to establish an asymptotic distribution for $\hat{\bm{\beta}}$. The first condition is technical and used to derive bounds on special functional classes used in the proof of asymptotic normality.  The other conditions are for convenience and ensure existence of limiting covariance matrices in the case of dependent data, and are similar to (C3) in \cite{LamarcheParker23}. 

The large sample theory of the estimator \eqref{eq:feqr} is described in the following result.

\begin{lemma}[\citet{GalvaoGuVolgushev20}] \label{lem:feqr}
 Under Assumptions \ref{assume:stationary}--\ref{a:boundedX}, and if $N/T^s \to 0$ for some $s \geq 1$, the estimator $\hat{\bm{\theta}}$ defined in \eqref{eq:feqr} is consistent for $\bm{\theta}_0$ as $N, T \to \infty$. Moreover, under Assumptions \ref{assume:stationary}--\ref{assume:Avar_dependent}, if $N (\log T)^4 / T \to 0$, then
  \begin{equation*}
    \sqrt{NT} (\hat{\bm{\beta}} - \bm{\beta}_0) \indist \mathcal{N}(\bm{0}, \bm{\Sigma}),
  \end{equation*}
  where
  \begin{equation} \label{eq:Sigma}
     \bm{\Sigma} = \bm{D}^{-1} \bm{V} \bm{D}^{-1}.
  \end{equation}
\end{lemma}

In Lemma~\ref{lem:feqr}, $T$ diverges to infinity to show the asymptotic results, and it must grow faster than $N$ for asymptotic normality.  The condition on the relative sizes of $T$ and $N$ is close to the standard rates for smooth non-linear panel data models, despite the lack of differentiability of the quantile regression score function $\psi_\tau(u) = \tau - I(u < 0)$.

We now offer some heuristics that help with intuition on the limitations of existing inference approaches applied to panel data quantile regression, as well as to establish the validity of the proposed bootstrap approach in Section \ref{sec:wbb} below. Our results rely on the Bahadur representation of the FE-QR estimator.  
\citet{Kato2012} established conditions under which the Bahadur representation of the estimator in equation \eqref{eq:feqr} holds, specifically
\begin{equation}\label{eq:bahadur}
\sqrt{NT}\left( \hat{\bm{\beta}} - \bm{\beta}_{0} \right) = \bm{D} _{N}^{-1}  \frac{1}{\sqrt{NT}} \sum_{i=1}^N \sum_{t=1}^T \tilde{\bm{x}}_{it} \psi _{\tau} \left( u_{it} \right) + \op{1},     
\end{equation} 
where $\tilde{\bm{x}}_{it}$, and $\bm{D}_N$ are defined in Assumption~\ref{assume:Avar_dependent}.  The other part of the variance that was defined in that assumption was labeled $\bm{V}_{NT}$ and is the variance of the remaining part of the linear term in~\eqref{eq:bahadur}.  With temporal dependence of unit $i$'s observations the variance of this linear term includes a weighted sum of covariances
\begin{equation}\label{eq:autoco}
    \gamma_{i}(k) = \ex{ \tilde{\bm{x}}_{it}  \psi _{\tau} \left( u_{it} \right) \tilde{\bm{x}}_{it+k}' \psi _{\tau} \left( u_{it+k} \right) }.
\end{equation}
Note that $\gamma_{i}(0) = \tau (1 - \tau) \ex{ \tilde{\bm{x}}_{it} \tilde{\bm{x}}_{it}'}$. Then we can write $\bm{V}_{NT}$ as
\begin{equation}
  \bm{V}_{NT} = \frac{1}{NT} \sum_{i=1}^N \left( T \gamma_i(0) + 2 \sum_{k=1}^{T-1} (T-k) \gamma_i(k) \right). \label{eq:def-Vn}
\end{equation}

We are interested in conducting inference for the parameter $\bm{\beta}_{0}$. Most straightforwardly, one could estimate the asymptotic variance-covariance matrix $\bm{\Sigma}$ in~\eqref{eq:Sigma} and construct confidence intervals directly from it. However, in the presence of serial dependence, a heteroskedasticity autocorrelation consistent (HAC) estimator of $\bm{V}_{NT}$ is essential, and at the moment, theoretical guidance about such an estimator is lacking in the literature. In addition, $\bm{D}_N$ depends on the conditional density of the error term, which is generally difficult to estimate well. Hence, the next section proposes a bootstrap procedure for inference on $\bm{\beta}_0$ in the FE-QR model under dependence. 

Before we provide details on the bootstrap procedure,  we introduce a required condition on the data generating process to establish the validity of the bootstrap described in Section~\ref{sec:wbb} below.
Let $\bm{V}_{NT}^0$ describe the variance of the sum in the Bahadur representation~\eqref{eq:bahadur}, that is, the first term in equation \eqref{eq:def-Vn}. In the case where errors at different time periods for each unit are independent, we can write \eqref{eq:def-Vn} as:
\begin{equation*}
    \bm{V}_{NT}^0 = \frac{1}{N} \sum_{i=1}^N \Var \left( \tilde{\bm{x}}_{i1} \psi_{\tau} \left( u_{i1} \right) \right) = \frac{1}{N} \sum_{i=1}^N \gamma_{i}(0).
\end{equation*}
Longitudinal data in the social sciences including economics commonly exhibit positive time-series dependence that makes the variance in~\eqref{eq:def-Vn} larger than $\bm{V}_{NT}^0$. In this paper, we will propose a bootstrap that is tailored to cases where
\begin{equation*}
    \bm{V}_{NT} - \bm{V}_{NT}^0 = \frac{2}{NT} \sum_{i=1}^N \sum_{k=1}^{T-1} (T-k) \gamma_i(k)
\end{equation*}
converges to a positive definite limit. This regularity condition is stated formally as Assumption~\ref{assump:A6} below.  To make the results more intuitive, we use the convention that for two conformable matrices $\bm{A}$ and $\bm{B}$, $\bm{A} \geq \bm{B}$ means that $\bm{A} - \bm{B}$ is positive semidefinite. Consider the following condition:

\begin{enumerate}[label=(A\arabic*), ref=(A\arabic*)]
\setcounter{enumi}{5}
  \item \label{assump:A6} Assume that $\lim_{N,T \rightarrow \infty} (\bm{V}_{NT} - \bm{V}_{NT}^0) \geq \zero$, and for each $k = 1, 2, \ldots$, 
  \begin{equation*}
    \ex{I(u_{it} < 0, u_{it+k} < 0) \tilde{\bm{x}}_{it} \tilde{\bm{x}}_{it+k}'} \leq \tau \ex{\tilde{\bm{x}}_{it} \tilde{\bm{x}}_{it+k}'}.
  \end{equation*}
\end{enumerate}

Assumption~\ref{assump:A6} imposes two conditions on the data for which the proposed bootstrap algorithm works well.  The first is on the sum of the autocovariances in \eqref{eq:def-Vn}, which allows for positive or negative dependence between different time periods. This condition together with the stationarity and mixing condition on the data in Assumption~\ref{assume:stationary} imply that the covariance between terms $\tilde{\bm{x}}_{it} \psi_\tau(u_{it})$ and $\tilde{\bm{x}}_{it+k} \psi_\tau(u_{it+k})$ can be bounded, due to the boundedness of the function $\psi_\tau(\cdot)$ that appears in the Bahadur representation~\eqref{eq:bahadur}.  As noted above, $\gamma_i(0)$ has a convenient familiar form.  The covariance between periods $t$ and $t+k$ can be similarly bounded.  Writing
\begin{equation*}
    \gamma_i(k) = \ex{\psi_\tau(u_{it}) \psi_\tau(u_{it+k}) \tilde{\bm{x}}_{it} \tilde{\bm{x}}_{it+k}'} = \ex{(I(u_{it} < 0, u_{it+k} < 0) - \tau^2) \tilde{\bm{x}}_{it} \tilde{\bm{x}}_{it+k}'},
\end{equation*}
the mixing condition in Assumption~\ref{assume:stationary} implies the bound $|\prob{u_{it} < 0, u_{it+k} < 0 \, | \, \bm{x}_i} - \tau^2| \rightarrow 0$, which means that it is less than $\tau(1 - \tau)$ for all $k$ large enough. This fact will be exploited by the partitioned wild bootstrap, which is described in detail in the next section.  The second condition in Assumption~\ref{assump:A6} limits the dependence between the regressors $\tilde{\bm{x}}_{it}$ and the quantile error terms.  Without such a condition, oscillating patterns that offset each other can cause problems with partition selection for the bootstrap.

\section{Partitioned wild bootstrap}\label{sec:wbb}

This section proposes a new wild residual bootstrap approach. For each unit $i$, we partition the time-series observations into $b$ non-overlapping parts of length $l$. We denote blocks of time series observations as partitions (or cells) instead of blocks to distinguish our approach from related methods which resample blocks \citep[see, e.g.,][among others]{Fitzenberger98,Goncalves11,GregoryLahiriNordman18,Hounyo23}. In contrast, we resample partition invariant weights to construct a bootstrap distribution.  The length $l$ of the partition is selected by a data-driven approach introduced below.

\subsection{Definition}

For simplicity, we assume that $T = b \times l$.  Letting $j = 1, 2, \ldots, b$, and denoting within-partition observations by $s = 1, 2, \ldots, l$, we redefine the FE-QR estimator in~\eqref{eq:feqr}:
\begin{equation}\label{eq:feqrboot0}
    \hat{\bm{\theta}} = (\hat{\bm{\beta}}', \hat{\bm{\alpha}}')' = \argmin_{ \bm{\theta} \in  \bm{\Theta} } \sum_{i=1}^N \sum_{j=1}^{b} \sum_{s=1}^l  \rho_\tau(y_{ijs} - \bm{x}_{ijs}' \bm{\beta} - \alpha_i).
\end{equation}
The corresponding residuals are $\hat{u}_{ijs} = y_{ijs} - \bm{x}_{ijs}' \hat{\bm{\beta}} - \hat{\alpha}_i$. New bootstrap observations are obtained from predicted values and reweighted residuals. Specifically, the bootstrap data generating process relies on the following bootstrap residuals
\begin{equation}\label{eq:wresid}
u_{ijs}^\ast = w_{ij} | \hat{u}_{ijs} |,
\end{equation}
where the weight $w_{ij}$ is drawn from a pre-specified distribution satisfying conditions \ref{b:wiid}--\ref{b:Gint} below. These within-partition invariant weights are independent and identically distributed (i.i.d.). Using the bootstrap residuals defined in~\eqref{eq:wresid}, the bootstrapped dependent variable is
\begin{equation*}
y_{ijs}^\ast = \bm{x}_{ijs}^{\prime} \hat{\bm{\beta}} + \hat{\alpha}_i + u_{ijs}^\ast.
\end{equation*}
Finally, a bootstrap FE-QR estimate is computed by finding
\begin{equation} \label{eq:feqrboot}
  \bm{\theta}^\ast = (\bm{\beta}^{\ast\prime}, \bm{\alpha}^{\ast\prime})^{\prime} = \argmin_{\bm{\theta} \in \bm{\Theta}} \sum_{i=1}^N  \sum_{j=1}^{b} \sum_{s=1}^{l}  \rho_\tau(y_{ijs}^\ast - \bm{x}_{ijs}^{\prime} \bm{\beta} - \alpha_i).
\end{equation}
This procedure is labeled partitioned wild bootstrap (PWB). Theorem \ref{thm:boot} below establishes the asymptotic validity of the method and Section \ref{sec:practical} describes how one can obtain valid confidence intervals using the bootstrap estimator proposed in \eqref{eq:feqrboot}. 

The distribution of the random weights is chosen by the researcher. As mentioned above, they are i.i.d., which makes them easy to generate, but must come from a distribution with CDF $G_W$ that satisfies the following conditions.

\begin{enumerate}[label=(B\arabic*), ref=(B\arabic*)]
    \item \label{b:wiid} The weights $\{ w_{ij}, 1 \leq i \leq N, 1 \leq j \leq b \}$ are independent and identically distributed.
    \item \label{b:Gquan} The $\tau$-th quantile of $G_W$ is $0$, that is, $G_W(0) = \tau$.
    \item \label{b:Gsupp} The support of $G_W$ is bounded and contained in $(-\infty, -c_1] \cup [c_2, \infty)$ for some $c_1, c_2 > 0$.
    \item \label{b:Gint} $G_W$ satisfies $-\int_{-\infty}^0 w^{-1} \ud G_W(w) = \int_0^\infty w^{-1} \ud G_W(w) = \frac{1}{2}$.
\end{enumerate}

Assumption \ref{b:wiid} is used in similar versions of the wild bootstrap with the exception that we require here that the weights are independent across partitions \citep[see, e.g.,][for a recent example]{LamarcheParker23}. Assumptions \ref{b:Gquan}--\ref{b:Gint} on $G_W$ were first proposed in \citet{FengHeHu11} and they are satisfied by several weight distributions \citep{FengHeHu11, lanWang2018}. 
We follow the literature and adopt a two-point mass distribution in the empirical examples. The distribution generates $w_{ij} = 2 (1-\tau)$ with probability $\tau$ and $w_{ij} = - 2 \tau$ with probability $(1-\tau)$.

\subsection{Practical implementation of the bootstrap}\label{sec:practical}
The practical implementation of the PWB method is simple. The main algorithm for implementing the methods is as follows. 

\begin{enumerate}[leftmargin=.25cm]
\item[] Step 1. For a given quantile of interest, fit the FE-QR panel model in equation \eqref{eq:feqrboot0} using the entire sample and compute the estimator $\hat{\bm{\beta}}$ and residuals $\hat{u}_{it}$;

\item[] Step 2. Select the partition size $l$ --- this is discussed in Section \ref{sec:celllength} below --- and make the number of cells $b=T/l$, with one shorter cell if necessary.  Let $j$ index partitions and relabel the observations $(\bm{x}_{ijs}, \hat{u}_{ijs})$ where $i$ indexes units, and $s$ indexes the time period within the $j$-th partition;

\item[] Step 3. Draw weights $\{w_{ij}\}$ for $1 \leq i \leq N$ and $1 \leq j \leq b$ randomly from distribution $G_W$ satisfying conditions \ref{b:wiid}--\ref{b:Gint}. Using the residuals from Step 1, compute bootstrap residuals  $u_{ijs}^\ast = w_{ij} | \hat{u}_{ijs} |$, and  $y_{ijs}^\ast = \bm{x}_{ijs}' \hat{\bm{\beta}} + \hat{\alpha}_i + u_{ijs}^\ast$;

\item[] Step 4. Using the sample, $(y_{ijs}^\ast , \bm{x}_{ijs})$ obtain the bootstrap estimator in equation \eqref{eq:feqrboot}. Denote the bootstrap estimator $\bm{\theta}^{\ast} = (\bm{\alpha}^{\ast}, \bm{\beta}^{\ast})$;

\item[] Step 5. Repeat Steps 3-4 $B$ times;

\item[] Step 6. Approximate the distribution of $\sqrt{NT} (\hat{\bm{\beta}} - \bm{\beta}_{0})$ by the empirical distribution of the $B$ observations of $\sqrt{NT}(\bm{\beta}^{\ast} - \hat{\bm{\beta}})$.

\end{enumerate}

By choosing the number of bootstrap simulations $B$ in the algorithm above large enough, the distribution of $\sqrt{NT}(\bm{\beta}^{\ast} - \hat{\bm{\beta}})$ can be computed with any desired precision.  There are several way of using this distribution for inference on the parameters.

\subsubsection*{Percentile confidence intervals}
The distribution function of $\bm{\beta}^\ast - \hat{\bm{\beta}}$ can be used to estimate the distribution function of $\hat{\bm{\beta}} - \bm{\beta}_0$.  Specifically, suppose that $\beta_0$ is one coordinate of $\bm{\beta}_0$ and we would like to compute a confidence interval for $\beta$ with confidence level $1-\lambda$.  Given bootstrap realizations $\{\beta_b^\ast\}_{b=1}^B$ of this coordinate, we can find $\beta_{\lambda/2}^\ast$ and $\beta_{1-\lambda/2}^\ast$.  Then one confidence interval for $\beta$ is
\begin{equation} \label{eq:boot_percentile}
\textnormal{CI}_P = [\beta_{\lambda / 2}^\ast, \beta_{1-\lambda / 2}^\ast].
\end{equation}
These percentiles may be used as to estimate the endpoints of a confidence interval for $\beta_0$.  

\subsubsection*{Variance-covariance matrix estimation and resulting confidence intervals}

For a fixed quantile level $\tau$, we define the bootstrap estimate of the asymptotic covariance matrix $\bm{\Sigma}$ given bootstrap realizations $\{\bm{\beta}^{\ast }_b\}_{b=1}^B$ as
\begin{equation}\label{eq:boot_var}
\bm{\Sigma}^{\ast} = \frac{1}{B} \sum_{b=1}^{B} (\bm{\beta}^{\ast}_b - \hat{\bm{\beta}})(\bm{\beta}^{\ast}_b - \hat{\bm{\beta}})^{\prime}.
\end{equation}
Under the regularity conditions in Theorem \ref{thm:boot} below, $\bm{\Sigma}^\ast$ is a consistent estimator of the asymptotic covariance matrix $\bm{\Sigma}$ defined in equation \eqref{eq:Sigma}.  The estimated standard errors of $\hat{\bm{\beta}}$ are the square roots of the diagonal elements of $\bm{\Sigma}^{\ast}$. Given $\bm{\Sigma}^\ast$, testing general hypotheses $R\bm{\beta}_0=r$ for the vector $\bm{\beta}_0$ can be accommodated by Wald-type tests.

In one dimension we can compare this method with the percentile method.  Once again, assume that $\beta_0$ is one coordinate of $\bm{\beta}_0$.  Let $\hat{\beta}$ be its estimate and let $se^\ast$ be the square root of the corresponding diagonal element of $\bm{\Sigma}^\ast$.  Then a standard-error based confidence interval is
\begin{equation}\label{eq:boot_se}
  \textnormal{CI}_{SE} = [\hat{\beta} - z_{1 - \lambda/2} se^\ast, \hat{\beta} - z_{1 - \lambda/2} se^\ast],
\end{equation}
where $z_{\lambda}$ denotes the $\lambda$-th quantile of the standard normal distribution. 

\subsection{Partition length selection} \label{sec:celllength}


The selection of the size of the partition is different for the PWB algorithm than for traditional block bootstrap approaches \citep[see, e.g.,][for examples of block bootstrap for inference on quantiles]{Fitzenberger98, Shao10ss, GregoryLahiriNordman18}. The length here is chosen to match the central term $\bm{V}$ in the asymptotic covariance matrix in equation~\eqref{eq:Sigma}, in particular the average of the covariances between non-contemporaneous score terms in the Bahadur representation~\eqref{eq:bahadur}, and the bootstrap variance.  The next result indicates how to choose the size of these cells, assuming that Assumptions~\ref{assume:stationary}--\ref{assump:A6} hold.

\begin{lemma}[Existence] \label{lem:l_exists}
Let Assumptions \ref{assume:stationary}--\ref{assump:A6} hold.  Define $\mathcal{V}_N: \mathbb{N} \rightarrow \R^{p \times p}$ by
\begin{equation*}
    \mathcal{V}_N(l) = \begin{cases} \mathbf{0} & l = 1 \\ \frac{2\tau(1-\tau)}{N} \sum_{i=1}^N \sum_{k=1}^{l-1} \frac{l - k}{l} \ex{\tilde{\bm{x}}_{it} \tilde{\bm{x}}_{it+k}'} & l = 2, 3, \ldots \end{cases}
\end{equation*}
Then, at least one $l^{\text{o}} \in \{1, 2, \ldots, T\}$ exists such that
\begin{equation} \label{lstar_choose}
    \lim_{N \rightarrow \infty} \mathcal{V}_N(l^{\text{o}})  \leq \lim_{N,T \rightarrow \infty} (\bm{V}_{NT} - \bm{V}_{NT}^0) \leq 
    \lim_{N \rightarrow \infty} \mathcal{V}_N(l^{\text{o}}+1).
\end{equation}
Furthermore, all such $l^{\text{o}}$ are uniformly bounded as $N, T \rightarrow \infty$.
\end{lemma}

The partition length $l^{\text{o}}$ in~\eqref{lstar_choose} can be estimated by numerically solving a finite-sample analog using plug-in estimates. Let $\check{\bm{x}}_{it} = \bm{x}_{it} - \bar{\varphi}_i^{-1} \bar{\bm{g}}_i$, where $\bar{\varphi}_i$ and $\bar{\bm{g}}_i$ are respectively the sample analogs of $\varphi_i$ and $\bm{g}_i$. Then we find $\hat{l}$ such that
\begin{multline} \label{eq:lhat_choose}
        \frac{1}{N} \sum_{i=1}^N \Bigg[  \tau(1-\tau) \sum_{k=1}^{\hat{l}-1} \frac{\hat{l} - k}{\hat{l}} \left( \frac{1}{b} \sum_{j=1}^b \frac{1}{\hat{l}-k} \sum_{s=1}^{\hat{l}-k} \check{\bm{x}}_{ijs} \check{\bm{x}}_{ijs+k}' \right) \\
    {} - \sum_{k=1}^{T-1} \frac{T-k}{T} K(k/h) \left( \frac{1}{T-k} \sum_{t=1}^{T-k} \check{\bm{x}}_{it} \psi_\tau(\hat{u}_{it}) \check{\bm{x}}_{it+k}' \psi_\tau(\hat{u}_{it+k}) \right) \Bigg] \approx 0,
\end{multline}
where $K(\cdot)$ is a kernel function for variance estimation with bandwidth $h$ to estimate each unit's contribution to the variance \citep[Section 3]{GalvaoYoon24}. Without kernel weighting, the right-hand side of \eqref{eq:lhat_choose} would be almost numerically zero -- these sums are sample analogs of the sum of $\gamma_i(k)$ terms for $k > 0$ on the right-hand side of~\eqref{eq:def-Vn}, but they also represent half of the off-diagonal terms of the first derivative of the empirical loss function evaluated at the optimizer, squared. To implement the selection rule, it is convenient to solve for the partition size unit by unit over the panel, leading to $\{\hat{l}_i\}_{i=1}^N$, where for each $i$,
\begin{multline}\label{lhats}
        \tau(1-\tau) \sum_{k=1}^{\hat{l}_i-1} \frac{\hat{l}_i - k}{\hat{l}_i} \left( \frac{1}{b} \sum_{j=1}^b \frac{1}{\hat{l}_i-k} \sum_{s=1}^{\hat{l}_i-k} \check{\bm{x}}_{ijs} \check{\bm{x}}_{ijs+k}' \right) \\
    {} - \sum_{k=1}^{T-1} \frac{T-k}{T} K(k/h) \left( \frac{1}{T-k} \sum_{t=1}^{T-k} \check{\bm{x}}_{it} \psi_\tau(\hat{u}_{it}) \check{\bm{x}}_{it+k}' \psi_\tau(\hat{u}_{it+k}) \right) \approx 0.
\end{multline}

Under the conditions of Lemma~\ref{lem:l_exists}, the feasible length selection $\hat{l}$ in \eqref{eq:lhat_choose} converges in probability to $l^\text{o}$, as stated in the following result.

\begin{lemma}[Consistency] \label{lem:lhat}
Under Assumptions \ref{assume:stationary}--\ref{assump:A6}, when $\hat{l}$ is chosen according to~\eqref{eq:lhat_choose}, $\hat{l} \pconv l^{\text{o}}$, for some $l^{\text{o}} \in \{1, \ldots, L\}$.
\end{lemma}

Length selection based on \eqref{eq:lhat_choose} is investigated in the simulation study reported in Section~\ref{sec:MC}. However, to help with the intuition on how the procedure works, we now offer an illustrative example. If we ignore the contributions of $\check{\bm{x}}_{it}$, $\hat{l}$ can be found to satisfy
\begin{equation*}
    \frac{1}{N} \sum_{i=1}^N \tau(1-\tau) \sum_{k=1}^{\hat{l}-1} \frac{\hat{l} - k}{\hat{l}} \approx \frac{1}{N} \sum_{i=1}^N \sum_{k=1}^{T-1} \frac{T-k}{T} K(k/h) \left( \frac{1}{T-k} \sum_{t=1}^{T-k} \psi_\tau(\hat{u}_{it}) \psi_\tau(\hat{u}_{it+k}) \right).
\end{equation*}
This approximate equality can be solved for $\hat{l}$, and, recalling that it should be a positive integer, we suggest
\begin{equation*}
    \hat{l} = 1 + \left\lceil \frac{2}{\tau (1-\tau)} \frac{1}{NT} \sum_{i=1}^N \sum_{k=1}^{T-1} K(k/h) \sum_{t=1}^{T-k} \psi_\tau(\hat{u}_{it}) \psi_\tau(\hat{u}_{it+k}) \right\rceil_+
\end{equation*}
where for $a \in \R$, $\lceil a \rceil$ refers to the smallest integer greater than $a$ (so we overestimate cell size slightly in small samples) and $a_+ = \max\{a, 0\}$, in case the term on the right is negative (under Assumption~\ref{assump:A6}, this occurs with probability decreasing to zero). Although the i.i.d. case is not considered in our investigation, under no temporal dependence when $T$ is sufficiently large, $\hat{l}$ in the last expression should be approximately equal to one, as expected.

\subsection{Bootstrap consistency}


As a direct implication of Lemma \ref{lem:l_exists}, the number of partitions is required to grow quickly as $T$ diverges to infinity. This is a minimal condition compared with block sizes conditions in the literature (see, e.g., C7 in \cite{GregoryLahiriNordman18} and Theorem 3.3 in \cite{Fitzenberger98}). The case of fixed block size is discussed in \citet{Fitzenberger98} but no formal results are given for quantile regression. Our approach relies on $b$ growing quickly, which is consistent with the large sample theory of FE-QR which is established under large $T$ panel approximations (Lemma \ref{lem:feqr}).


The next result describes the validity of our bootstrap approach.

\begin{theorem} \label{thm:boot}
Under the conditions of Lemma \ref{lem:l_exists} and Assumptions \ref{b:wiid}--\ref{b:Gint}, and assuming $N (\log T)^4 / T = o(1)$ as $N, T \to \infty$,
  \begin{equation*}
    \sup_{\bm{\upsilon} \in \RR^p} \left| \prob{\sqrt{NT}( \bm{\beta}^* - \hat{\bm{\beta}}) \leq \bm{\upsilon} | \bm{S}} - \prob{\sqrt{NT}( \hat{\bm{\beta}} - \bm{\beta}_0) \leq \bm{\upsilon}} \right| \inpr 0
  \end{equation*}
  where $\bm{S}$ denotes the observed sample and $\bm{\beta}^\ast$ denotes the slope estimator defined by \eqref{eq:feqrboot}.
\end{theorem}

The result in Theorem \ref{thm:boot} implies the consistency of the PWB for the FE-QR estimator in the case of dependent data. In the next sections, we estimate the partition length $l^{\text{o}}$ and document the performance of the feasible version of the proposed PWB estimator.

\begin{rem}
The bootstrap variance consistently estimate the asymptotic variance when the sum of the autocovariances in \eqref{eq:def-Vn} is positive semidefinite, which allows for positive or negative dependence between different time periods. However, in empirical applications where negative autocovariances dominate the variance expression, in conflict with Assumption \ref{assump:A6}, the bootstrap variance will converge to the variance obtained under i.i.d. conditions, overestimating the asymptotic variance. Thus, our approach could offer practitioners conservative statistical inference.
\end{rem}


Under moment inequalities for mixing processes and the finiteness of $l^\text{o}$, the argument for consistent i.i.d. variance estimation in \citet{LamarcheParker23} implies the finiteness of this variance estimator, with the resulting consistent/conservative description depending on whether the data obey Assumption~\ref{assump:A6}.  It is interesting to note that conservative bootstrap inference due to nonexistent bootstrap moments was recently considered in~\citet{HahnLiao21}. Also, \citet{MachadoParente5} propose an $L$-estimator of the variance that might apply under more general data conditions than those that we consider, but we leave investigation of such an estimator to future research.

\section{Simulation results}\label{sec:MC}

In this section, we report results of several simulation experiments designed to evaluate the finite sample performance of the proposed method. We consider a data generating process similar to the one considered in \cite{GalvaoGuVolgushev20}. The dependent variable is $y_{it} = \alpha_i + x_{it} + (1 + \zeta x_{it}) u_{it}$, where $x_{it} = 0.5 \alpha_i + z_i + \epsilon_{it}$, and $z_i$ is an independent and identically distributed (i.i.d.) random variable distributed as $\chi^2$ with 3 degrees of freedom ($\chi_3^2$). The corresponding quantile regression function is
$Q_y (\tau | x_{it}) = \alpha_{i0} + \beta_0 x_{it}$, where $\alpha_{i0} = \alpha_i + F^{-1}(\tau)$, $\beta_0 = 1 + \zeta F^{-1}(\tau)$, and $F(\cdot)$ denotes the CDF of $u_{it}$. By varying $\zeta \in \{0,0.25\}$, we are able to generate data from two variations of the basic model. The location-shift model assumes  $\zeta=0$ and then $\beta_0 = 1$. In the location-scale shift version of the model, $\zeta = 0.25$, and $\beta_0 = 1 + 0.25 F^{-1}(\tau)$ varies by quantile $\tau$.

We generate the observations by combining different distributions for $\alpha_i$ and $x_{it}$. As in \cite{LamarcheParker23}, we assume that the individual intercept $\alpha_i=i/N$ for $1 \leq i \leq N$, or alternatively, $\alpha_i$ is assumed to be an i.i.d. Gaussian random variable.

We depart from the simulation study in \cite{GalvaoGuVolgushev20} by allowing the errors to be serially dependent:
\begin{equation*}
    u_{it} = \rho_{1,u} u_{it-1}  + \rho_{2,u} u_{it-2} + v_{it,u}, \; \; \mbox{and} \; \; \epsilon_{it} = \rho_{1,\epsilon} \epsilon_{it-1}  + \rho_{2,\epsilon} \epsilon_{it-2} + v_{it,\epsilon}.
\end{equation*}
We consider innovation terms $v_{it,u}, v_{it,\epsilon}$ following $\mathcal{N}(0,1)$ and $t_3$ distributions and set the auto-regressive parameters $\rho_{1,u} = \rho_{1,\epsilon}=0.7$ and $\rho_{2,u} = \rho_{2,\epsilon}=0.1$ as in \cite{GregoryLahiriNordman18}.  The variance 
$\sigma_u^2 = (1 - \rho_{2,u}) \sigma_v^2 / ((1 + \rho_{2,u}) (1 - \rho_{1,u} - \rho_{2,u}) ( 1 + \rho_{1,u} - \rho_{2,u} ))$, which is approximately $2.56$ for normal innovations and three times that for $t_3$ innovations.

\begin{figure}
\begin{center}
\centerline{\includegraphics[width=1\textwidth]{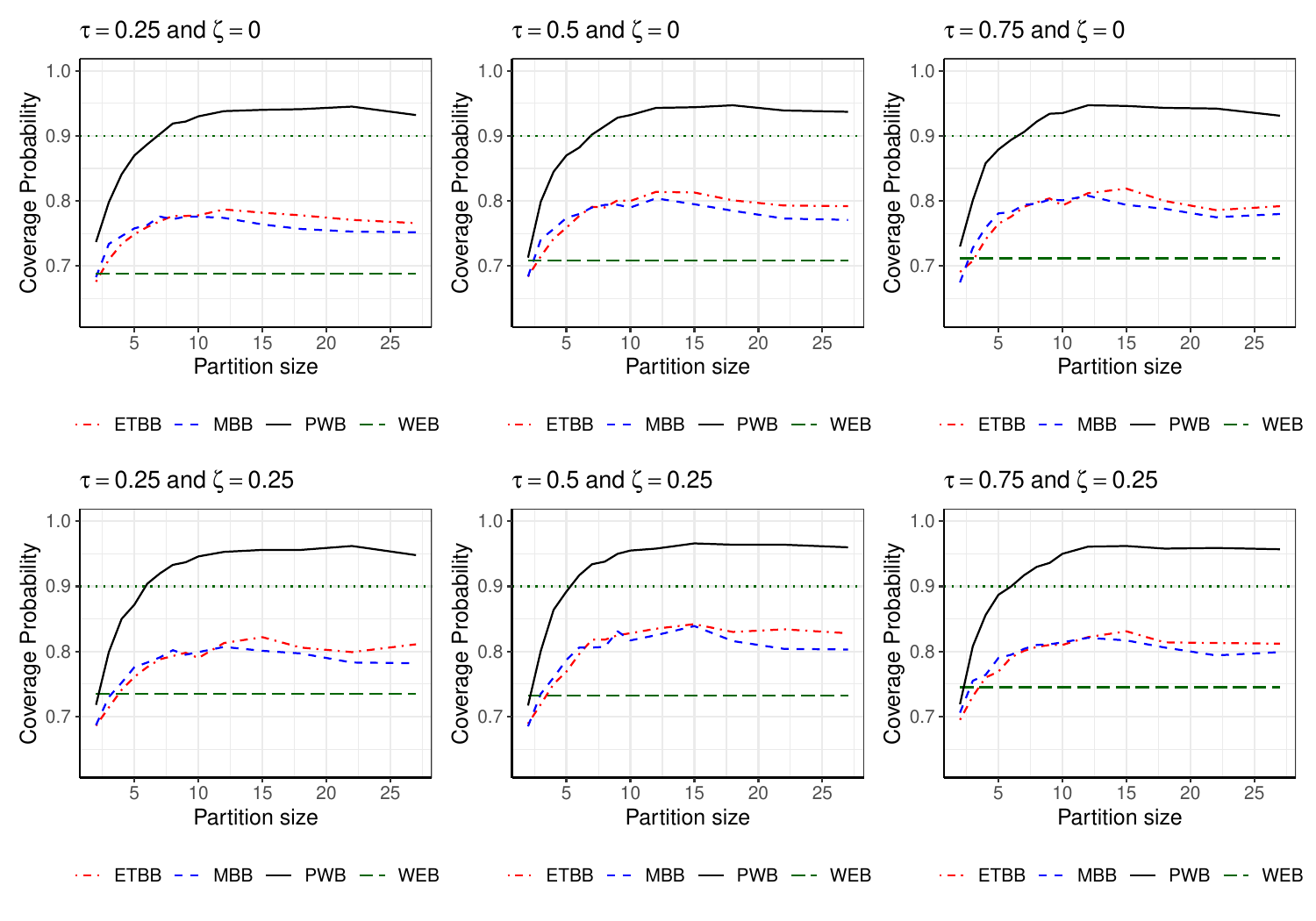}}
\caption{\emph{Coverage probability for a nominal level of 90\% when $N=5$ and $T=200$. The percentile confidence interval $\textnormal{CI}_P$ is constructed as in \eqref{eq:boot_percentile}. MBB denotes moving block bootstrap, ETBB extended tapered block bootstrap, WEB multiplier exponential block bootstrap, and PWB partitioned wild bootstrap.}\label{mc:fig1}}
\end{center}
\end{figure}

\begin{figure}
\begin{center}
\centerline{\includegraphics[width=1\textwidth]{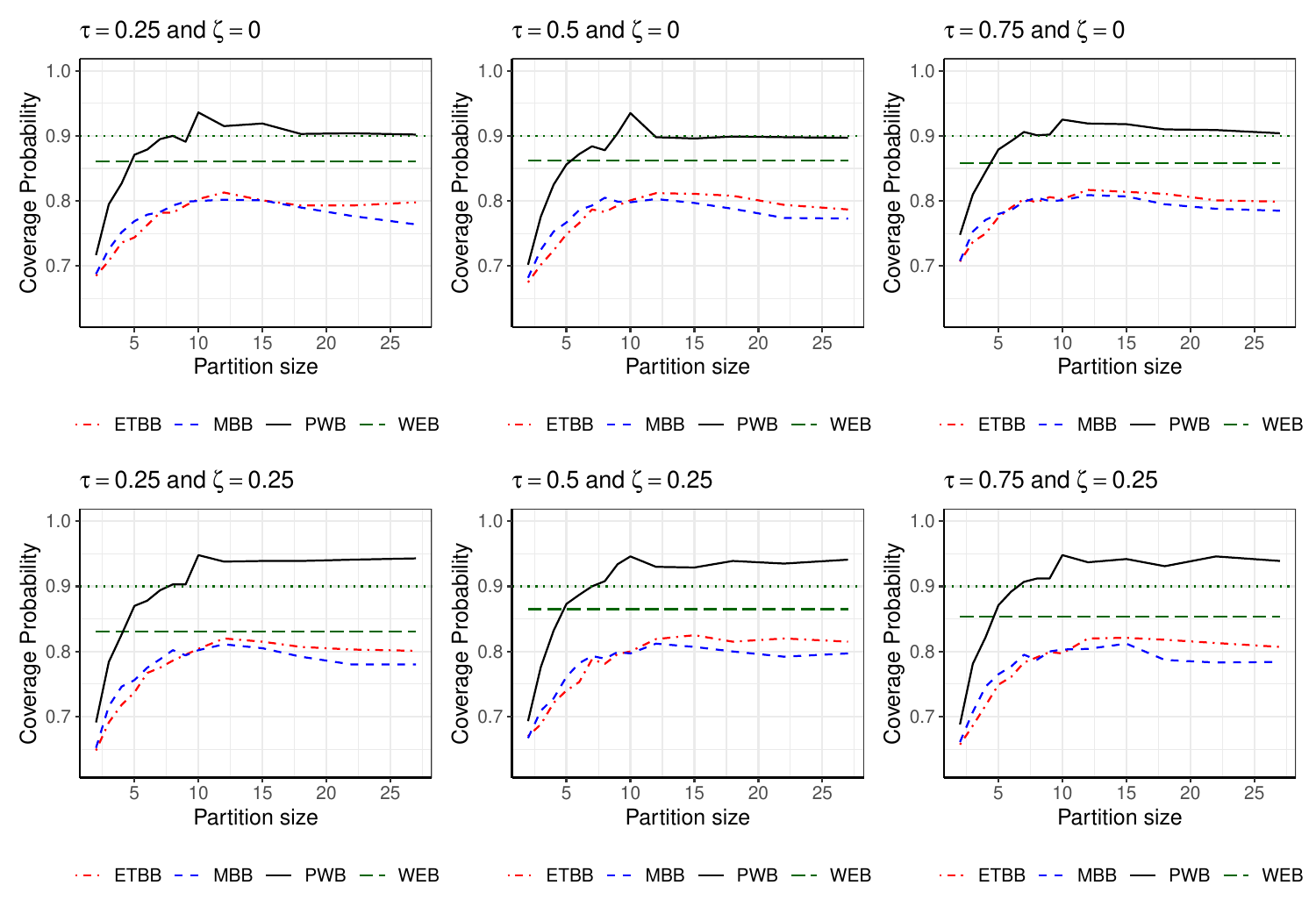}}
\caption{\emph{Coverage probability for a nominal level of 90\% when $N=50$ and $T=200$. The percentile confidence interval $\textnormal{CI}_P$ is constructed as in \eqref{eq:boot_percentile}. MBB denotes moving block bootstrap, ETBB extended tapered block bootstrap, WEB multiplier exponential block bootstrap, and PWB partitioned wild bootstrap.}\label{mc:fig2}}
\end{center}
\end{figure}

 Lastly, we use different combinations of $N \in \{5,50\}$ and $T\in \{200,400\}$ and quantile levels $\tau \in \{0.25,0.50,0.75\}$. The results are obtained by using 1000 random samples and 400 bootstrap repetitions. 
 
Figures \ref{mc:fig1} and \ref{mc:fig2} present results assuming that $\alpha_i = i/N$ and $u_{it} \sim \mathcal{N}(0,\sigma_u^2)$. The figures show coverage probabilities for a nominal level of 90\% for the slope parameter over a range of values for the size of the partition $l$. The intervals are constructed using the empirical bootstrap distribution following \eqref{eq:boot_percentile}, and consequently, we present coverages obtained using bootstrap methods. The figures show the performance of the moving block bootstrap (MBB) proposed by~\citet{Fitzenberger98}, extended tapered block bootstrap (ETBB) of~\citet{Shao10ss} generalized to quantile regression as in \citet{GregoryLahiriNordman18}, weighted block bootstrap (WEB) of~\citet{GalvaoParkerXiao24}, and the proposed PWB estimator. The first three procedures are briefly described in the online appendix. 
While the weights vary by block size in the MBB, ETBB and PWB approaches, the weight to each unit $i$ is constant in the case of WEB. The figures present results for the location-shift case in the upper panels, and the location-scale shift model in the lower panels. To understand the effect of temporal dependence in panel quantiles, we start with $N=5$ in Figure \ref{mc:fig1}, and then we increase $N$ to 50 to generate the results shown in Figure \ref{mc:fig2}. 

Figure \ref{mc:fig1} shows that the coverage of the MBB and ETBB improve as we increase the partition size from $l = 2$, although it remains roughly constant when $l \geq 10$. We observe the same pattern for PWB, but the method reaches levels closer to the target 90\% when $5 \leq l \leq 7$. The performance of WEB is not surprising, since the method relies on cross-sectional variation. We note that when $N$ increases, as in Figure \ref{mc:fig2}, the performance of WEB significantly improves and it is superior to MBB and ETBB. On the other hand, PWB offers the best coverage in all cases and the excellent performance of the approach is not restricted to a single cell size. 

\begin{table}
\begin{center}
\begin{tabular}{c c c c c c c c c c c c } \hline
        &   &  & \multicolumn{4}{c}{Bootstrap Methods} &  & \multicolumn{4}{c}{Bootstrap Methods} \\ 
$N,T$ 	&	$\tau$ 	&	\small{PO}	&	\small{MBB}	&	\small{ETBB}	&	\small{WEB}	&	\small{PWB}	&	\small{PO}	&	\small{MBB}	&	\small{ETBB}	&	\small{WEB}	&	\small{PWB}	\\   \hline
    & &  \multicolumn{10}{c}{Location shift model ($\zeta=0$)} \\
            &   &  \multicolumn{5}{c}{$\alpha_i = i/N$}  &  \multicolumn{5}{c}{$\alpha_i \sim \mathcal{N}(0,1)$} \\   \hline
5,200	&	0.25	&	0.672	&	0.745	&	0.744	&	0.756	&	0.869	&	0.673	&	0.762	&	0.757	&	0.761	&	0.895	\\
	&	0.50	&	0.698	&	0.783	&	0.744	&	0.746	&	0.865	&	0.673	&	0.746	&	0.724	&	0.743	&	0.870	\\
	&	0.75	&	0.652	&	0.752	&	0.728	&	0.738	&	0.869	&	0.657	&	0.760	&	0.742	&	0.752	&	0.882	\\
50,200	&	0.25	&	0.647	&	0.764	&	0.753	&	0.884	&	0.898	&	0.629	&	0.760	&	0.719	&	0.852	&	0.876	\\
	&	0.50	&	0.657	&	0.772	&	0.733	&	0.883	&	0.876	&	0.611	&	0.749	&	0.719	&	0.884	&	0.876	\\
	&	0.75	&	0.660	&	0.739	&	0.733	&	0.873	&	0.886	&	0.662	&	0.775	&	0.745	&	0.879	&	0.893	\\
5,400	&	0.25	&	0.649	&	0.783	&	0.751	&	0.750	&	0.905	&	0.637	&	0.753	&	0.715	&	0.750	&	0.891	\\
	&	0.50	&	0.622	&	0.752	&	0.697	&	0.727	&	0.881	&	0.644	&	0.767	&	0.709	&	0.746	&	0.892	\\
	&	0.75	&	0.649	&	0.767	&	0.720	&	0.740	&	0.882	&	0.655	&	0.759	&	0.727	&	0.735	&	0.906	\\
50,400	&	0.25	&	0.611	&	0.754	&	0.718	&	0.871	&	0.888	&	0.673	&	0.782	&	0.760	&	0.871	&	0.898	\\
	&	0.50	&	0.598	&	0.742	&	0.679	&	0.865	&	0.883	&	0.657	&	0.793	&	0.748	&	0.891	&	0.911	\\
	&	0.75	&	0.620	&	0.769	&	0.732	&	0.874	&	0.891	&	0.664	&	0.795	&	0.765	&	0.898	&	0.919	\\   \hline
         & &  \multicolumn{10}{c}{Location-scale shift model ($\zeta=0.25$)} \\ 
            &   &  \multicolumn{5}{c}{$\alpha_i = i/N$}  &  \multicolumn{5}{c}{$\alpha_i \sim \mathcal{N}(0,1)$} \\    \hline
5,200	&	0.25	&	0.704	&	0.788	&	0.768	&	0.786	&	0.920	&	0.694	&	0.755	&	0.744	&	0.761	&	0.898	\\
	&	0.50	&	0.717	&	0.795	&	0.750	&	0.805	&	0.902	&	0.714	&	0.762	&	0.728	&	0.777	&	0.893	\\
	&	0.75	&	0.711	&	0.781	&	0.755	&	0.772	&	0.914	&	0.703	&	0.778	&	0.754	&	0.770	&	0.889	\\
50,200	&	0.25	&	0.626	&	0.742	&	0.711	&	0.869	&	0.875	&	0.606	&	0.723	&	0.708	&	0.853	&	0.880	\\
	&	0.50	&	0.652	&	0.770	&	0.725	&	0.887	&	0.893	&	0.638	&	0.751	&	0.713	&	0.880	&	0.882	\\
	&	0.75	&	0.622	&	0.747	&	0.724	&	0.880	&	0.901	&	0.608	&	0.722	&	0.682	&	0.856	&	0.872	\\
5,400	&	0.25	&	0.652	&	0.750	&	0.710	&	0.766	&	0.901	&	0.663	&	0.768	&	0.719	&	0.749	&	0.903	\\
	&	0.50	&	0.675	&	0.758	&	0.697	&	0.753	&	0.895	&	0.690	&	0.764	&	0.717	&	0.759	&	0.900	\\
	&	0.75	&	0.668	&	0.779	&	0.724	&	0.776	&	0.908	&	0.707	&	0.792	&	0.744	&	0.785	&	0.905	\\
50,400	&	0.25	&	0.625	&	0.731	&	0.715	&	0.882	&	0.908	&	0.608	&	0.746	&	0.679	&	0.851	&	0.888	\\
	&	0.50	&	0.649	&	0.755	&	0.723	&	0.886	&	0.906	&	0.630	&	0.764	&	0.707	&	0.869	&	0.892	\\
	&	0.75	&	0.645	&	0.776	&	0.716	&	0.889	&	0.915	&	0.597	&	0.721	&	0.662	&	0.848	&	0.885	\\  \hline  
\end{tabular}
\end{center}
\caption{\emph{Rejection probability for a nominal level of 90 percent under correlated errors distributed as Normal. The standard-error based confidence interval $\textnormal{CI}_{SE}$ is constructed as in \eqref{eq:boot_var}. PO is the Powell Kernel estimator, MBB denotes moving block bootstrap, ETBB extended tapered block bootstrap, WEB weighted exponential block bootstrap, and PWB partitioned wild bootstrap.}}
\label{table1.mc.results}
\end{table}

\begin{table}
\begin{center}
\begin{tabular}{c c c c c c c c c c c c } \hline
        &   &  & \multicolumn{4}{c}{Bootstrap Methods} &  & \multicolumn{4}{c}{Bootstrap Methods} \\ 
$N,T$ 	&	$\tau$ 	&	\small{PO}	&	\small{MBB}	&	\small{ETBB}	&	\small{WEB}	&	\small{PWB}	&	\small{PO}	&	\small{MBB}	&	\small{ETBB}	&	\small{WEB}	&	\small{PWB}	\\   \hline
    & &  \multicolumn{10}{c}{Location shift model ($\zeta=0$)} \\
            &   &  \multicolumn{5}{c}{$\alpha_i = i/N$}  &  \multicolumn{5}{c}{$\alpha_i \sim \mathcal{N}(0,1)$} \\   \hline
5,200	&	0.25	&	0.651	&	0.759	&	0.746	&	0.777	&	0.872	&	0.626	&	0.750	&	0.710	&	0.729	&	0.865	\\
	&	0.50	&	0.679	&	0.780	&	0.737	&	0.758	&	0.877	&	0.636	&	0.742	&	0.705	&	0.702	&	0.844	\\
	&	0.75	&	0.641	&	0.772	&	0.735	&	0.783	&	0.891	&	0.624	&	0.758	&	0.715	&	0.744	&	0.856	\\
50,200	&	0.25	&	0.642	&	0.763	&	0.732	&	0.887	&	0.888	&	0.660	&	0.767	&	0.741	&	0.875	&	0.887	\\
	&	0.50	&	0.616	&	0.757	&	0.721	&	0.872	&	0.855	&	0.632	&	0.747	&	0.734	&	0.892	&	0.886	\\
	&	0.75	&	0.651	&	0.777	&	0.742	&	0.883	&	0.875	&	0.646	&	0.764	&	0.734	&	0.866	&	0.875	\\
5,400	&	0.25	&	0.633	&	0.773	&	0.721	&	0.754	&	0.890	&	0.643	&	0.769	&	0.733	&	0.751	&	0.882	\\
	&	0.50	&	0.646	&	0.752	&	0.708	&	0.745	&	0.869	&	0.644	&	0.750	&	0.708	&	0.725	&	0.868	\\
	&	0.75	&	0.622	&	0.757	&	0.719	&	0.748	&	0.885	&	0.640	&	0.754	&	0.720	&	0.740	&	0.878	\\
50,400	&	0.25	&	0.646	&	0.778	&	0.728	&	0.888	&	0.897	&	0.631	&	0.785	&	0.746	&	0.889	&	0.903	\\
	&	0.50	&	0.643	&	0.796	&	0.733	&	0.900	&	0.902	&	0.614	&	0.753	&	0.700	&	0.879	&	0.883	\\
	&	0.75	&	0.649	&	0.796	&	0.753	&	0.898	&	0.904	&	0.611	&	0.758	&	0.729	&	0.874	&	0.885	\\  \hline
         & &  \multicolumn{10}{c}{Location-scale shift model ($\zeta=0.25$)} \\ 
            &   &  \multicolumn{5}{c}{$\alpha_i = i/N$}  &  \multicolumn{5}{c}{$\alpha_i \sim \mathcal{N}(0,1)$} \\    \hline
5,200	&	0.25	&	0.694	&	0.780	&	0.754	&	0.790	&	0.889	&	0.657	&	0.743	&	0.725	&	0.736	&	0.877	\\
	&	0.50	&	0.715	&	0.764	&	0.739	&	0.792	&	0.879	&	0.708	&	0.782	&	0.742	&	0.801	&	0.896	\\
	&	0.75	&	0.686	&	0.771	&	0.740	&	0.792	&	0.891	&	0.686	&	0.787	&	0.735	&	0.781	&	0.896	\\
50,200	&	0.25	&	0.598	&	0.725	&	0.681	&	0.850	&	0.859	&	0.590	&	0.709	&	0.694	&	0.854	&	0.856	\\
	&	0.50	&	0.613	&	0.774	&	0.698	&	0.878	&	0.886	&	0.668	&	0.765	&	0.736	&	0.891	&	0.888	\\
	&	0.75	&	0.627	&	0.743	&	0.710	&	0.865	&	0.880	&	0.609	&	0.744	&	0.694	&	0.871	&	0.875	\\
5,400	&	0.25	&	0.649	&	0.782	&	0.748	&	0.786	&	0.902	&	0.681	&	0.773	&	0.737	&	0.777	&	0.892	\\
	&	0.50	&	0.694	&	0.776	&	0.747	&	0.780	&	0.889	&	0.708	&	0.783	&	0.731	&	0.781	&	0.886	\\
	&	0.75	&	0.663	&	0.789	&	0.727	&	0.763	&	0.897	&	0.631	&	0.756	&	0.694	&	0.767	&	0.884	\\
50,400	&	0.25	&	0.639	&	0.768	&	0.717	&	0.877	&	0.894	&	0.566	&	0.724	&	0.666	&	0.864	&	0.874	\\
	&	0.50	&	0.629	&	0.781	&	0.729	&	0.894	&	0.900	&	0.640	&	0.766	&	0.723	&	0.886	&	0.885	\\
	&	0.75	&	0.591	&	0.727	&	0.667	&	0.849	&	0.868	&	0.572	&	0.715	&	0.676	&	0.847	&	0.865	\\  \hline  
\end{tabular}
\end{center}
\caption{\emph{Rejection probability for a nominal level of 90 percent under correlated errors distributed as $t_3$. The standard-error based confidence interval $\textnormal{CI}_{SE}$ is constructed as in \eqref{eq:boot_var}. PO is the Powell Kernel estimator, MBB denotes moving block bootstrap, ETBB extended tapered block bootstrap, WEB weighted exponential block bootstrap, and PWB partitioned wild bootstrap.}}
\label{table2.mc.results}
\end{table}

The evidence reveals that by judiciously selecting the partition size, the empirical coverage of the PWB method can reach the target 90\% coverage level. Naturally, the best partition size is unknown, so we now extend our investigation to the feasible version of the estimator. Tables \ref{table1.mc.results} and \ref{table2.mc.results} present coverage probabilities for a nominal level of 90\% obtained by estimating the variance of the estimator as in \eqref{eq:boot_var}. Table \ref{table1.mc.results} presents results for errors distributed as Normal, and Table \ref{table2.mc.results} presents results for errors distributed as $t_3$. Relative to Figures \ref{mc:fig1} and \ref{mc:fig2}, we are able to expand the evidence to consider different sample sizes, alternative methods for inference on the parameters (i.e., $\textnormal{CI}_{P}$ vs $\textnormal{CI}_{SE}$), and different distributions. The first column shows coverage of the Powell's estimator (PO) of the asymptotic covariance matrix of the FE-QR estimator. PO has been shown to be consistent under the assumption that the data is i.i.d. (Proposition 3.1 in \cite{Kato2012}). In the next columns, we present the bootstrap approaches. For the MBB and ETBB, we use the {\tt R} package {\tt QregBB} for determining the block length as discussed in \cite{GregoryLahiriNordman18}. In the case of PWB, coverage is obtained for the selected partition length $\hat{l}$ (summaries are shown in Figure \ref{mc:fig3} below). 

In all the variations of the model considered in Table \ref{table1.mc.results}, the PWB estimator performs better than the other estimators. Considering for instance the location-shift model, PO does not perform well, although this is expected because the estimator does not account for the temporal dependence of the errors. The weights used by the MBB and ETBB methods are time-varying because they are primarily intended to be used with time-series data.  They are not very well suited for panel data because of the inclusion of time-invariant indicators of unit membership~--- there is no transformation of the data that can be made beforehand to remove individual effects.  As a result, the MBB and ETBB methods undercover. WEB performs well only when $N$ is large relative to $T$, as expected from the evidence presented in Figure \ref{mc:fig1}. 

The results for the location-scale shift model presented in the second panel in Table \ref{table1.mc.results} are similar. We continue to see that PWB performs better than alternative methods. This conclusion holds when we consider the evidence in Table \ref{table2.mc.results}. The evidence confirms two results. The PWB offers better coverage than WEB in models with $N$ small relative to $T$ and the procedure proposed in this paper is valid for approximating the distribution of the FE-QR estimator. 

\begin{figure}
\begin{center}
\centerline{\includegraphics[width=1\textwidth]{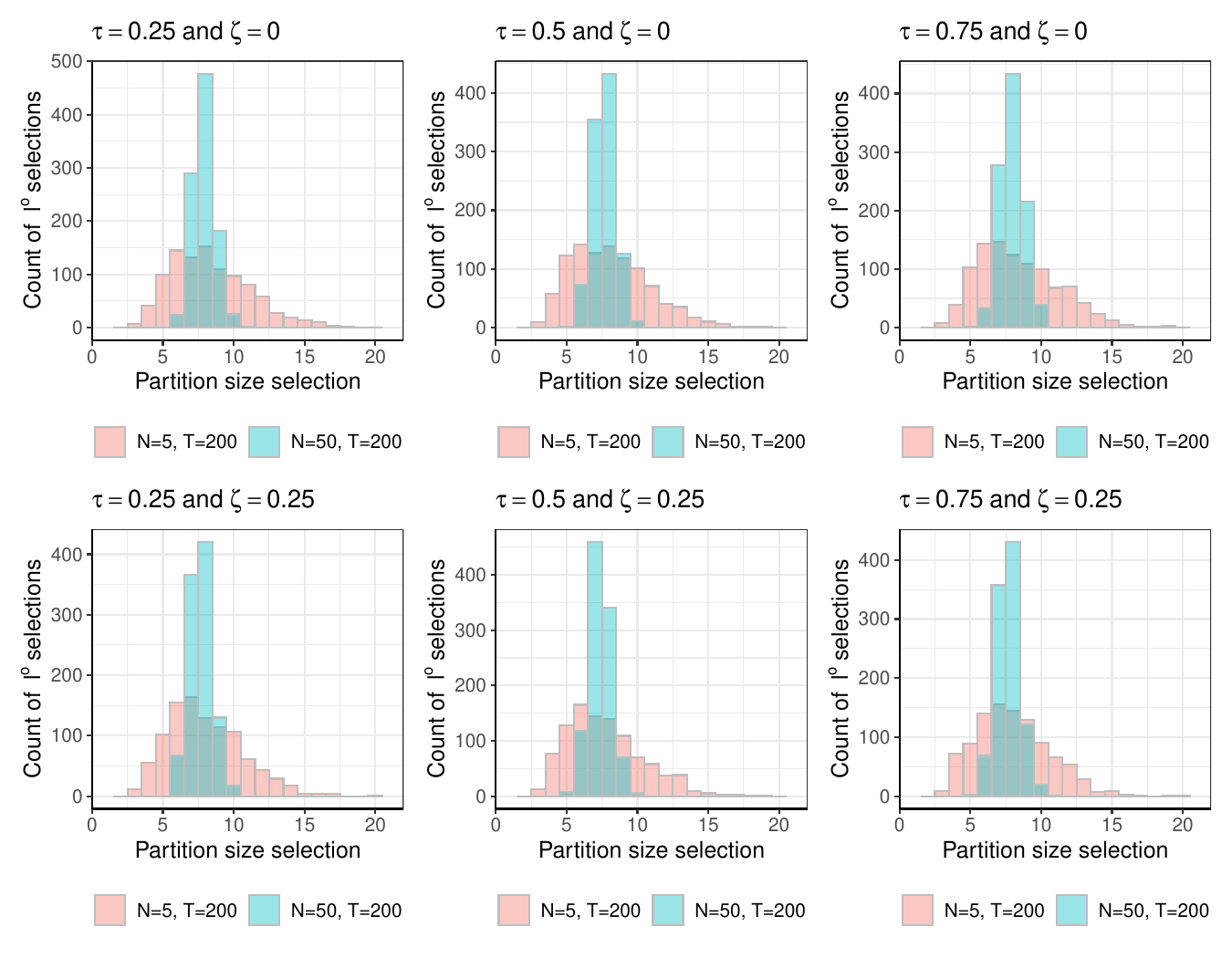}}
\caption{\emph{Count of partition size selections for the PWB. The figure reports selected length size over 1000 samples corresponding to the results in Table \ref{table1.mc.results} when $\alpha_i = i/N$.}\label{mc:fig3}}
\end{center}
\end{figure}


Finally, using Figure \ref{mc:fig3}, we present the frequency of partition size selections for PWB corresponding to $N \in \{5,50\}$, $T=200$ and $\alpha_i = i/N$ in Table \ref{table1.mc.results}. The frequencies of $\hat{l}$ corresponding to the other entries in Tables \ref{table1.mc.results} and \ref{table2.mc.results} are similar and we do not report them to save space. We first estimate the residuals $\hat{u}_{it}$ using the FE-QR estimator and construct $\check{x}_{it} \psi_\tau(\hat{u}_{it}) \check{x}_{it+k} \psi_\tau(\hat{u}_{it+k})$ for all $k$ considering $\check{x}_{it} = x_{it} - \bar{x}_i$, where $\bar{x}_i = \sum_{t=1}^T x_{it}/T$. This choice does not involve estimation of nuisance parameters and performed well in the simulations. We also employed a triangular kernel $(1-|k|/h)$ with $h=1$. For each $1 \leq i \leq N$, equation \eqref{lhats} was evaluated at each length $l$ in the contained set $\{1,2,\hdots,L\}$, where $L$ was set to 25 to minimize computational time. We then selected the partition size $\hat{l}_i$ that best fit $\eqref{lhats}$. Finally, we obtained $\hat{l}$ by averaging over individual partition sizes.

The excellent performance of PWB in Tables \ref{table1.mc.results} and \ref{table2.mc.results} is conditional on estimates of the partition size obtained from this selection rule, and thus, the evidence confirms that the variance can be approximated empirically by selecting the length of the partition. As expected based on the evidence reported in Figure \ref{mc:fig1}, Figure \ref{mc:fig3} shows that the highest relative frequencies are obtained for values of $\hat{l}$ between 6 and 9. Moreover, we observe that the precision of the selection rule quickly increases when $N$ increases. For instance, the figure shows that over 70\% of the selections are in the range $7 \leq \hat{l} \leq 8$ when $N=50$, consistent with the values of $l$ that produce coverage probabilities close to the nominal level of 90\% in Figure \ref{mc:fig2}.

\section{Empirical illustration}\label{sec:application}

In this section, we illustrate the use of the proposed approach by employing data from a randomized control trial on electricity consumption. We estimate a panel data quantile regression model and apply the method to evaluate how consumers respond to time-of-use  electricity pricing. We use a data set that includes $N = 268$ customers observed over $T = 2,160$ time intervals of 30 minutes. Despite the increasing number of studies using household-level panel data with large $T$ \citep[see][among others]{kJessoe2014,mHarding2016}, dependence within household has been ignored in the empirical literature on quantile regression. Our findings suggest that households do not respond to a modest change in the price of electricity that occurs during the day. However, when the price of electricity increases by 85 percent during the evening, high-usage households reduce their consumption by about 10 percent. If we compare existing approaches with the proposed approach, we find that practitioners would tend to conclude incorrectly that the effects of modest and large changes in the price of electricity are significant across the conditional distribution if they ignore the positive temporal dependence in each household's observations.

\subsection{Data and model}

We use the CER Smart Metering Project data from the Irish Social Science Data Archive (ISSDA). The experiment was conducted from 2008 to 2011 and we employ data from the period first two months of 2010. Electricity consumption was recorded over 30 minute intervals from Monday to Friday at the residential level. We consider only two treatment types in this paper. The customers selected for the control group had a time-invariant rate of \EUR{0.141} per kilowatt hour (kwh). Customers selected for the treatment group were charged \EUR{0.135} per kwh (Day), with the exception of \EUR{0.110} from 23:00 to 8:00 (Night) and \EUR{0.260} from 17:00 to 19:00 (Peak). Households in the treatment group ($N_1=68$) received an in-home display (IHD) device and an energy usage statement in their bills. Households in the control group ($N_0=200$) did not receive an IHD. 

We estimate the following panel data model:    
\begin{equation}
y_{it} = \alpha_{i} + \beta_1 d_{t,1} + \beta_2 d_{t,2} + \beta_3 (d_i \times d_{t,1}) + \beta_4 (d_i \times d_{t,2}) +  \bm{x}_{it}' \bm{\delta} + u_{it},  \label{mainE}
\end{equation} 
where $y_{igt}$ is the natural logarithm of electricity usage for household $i$ during the  interval $t$. The variable $d_i$ indicates treatment status and its effect is not identified because the model includes household fixed effects, $\alpha_i$. We can identify, however, the effect of being treated at different times during the day. The variable $d_{t,1}$ denotes whether $t \in [8-17] \cup [19-23]$ and $d_{t,2}$ denotes if $t \in [17-19]$. The coefficients of interest $(\beta_4(\tau), \beta_5(\tau))$ are identified by the time variation associated with time-of-use pricing across households.

The vector of independent variables $\bm{x}_{it}$ includes average temperature and average relative humidity in Ireland, which are household-invariant regressors. We also include the logarithm of electricity consumption 30 minutes earlier and the logarithm of electricity consumption a day earlier, which accounts for the correlation of electricity usage across weekdays due to daily routines. We have household size, size of the house, and other characteristics of the household and home in the CER data, but we do not include them in the model because these variables are time-invariant and the model includes household fixed effects. 

\begin{singlespace}
\begin{table}
\begin{small}
\begin{center}
\begin{tabular}{l c c c c c c  }
\hline
\multicolumn{1}{l}{} & \multicolumn{5}{c}{Quantile Regression} & \multicolumn{1}{c}{Mean} \\ 
\multicolumn{1}{l}{} & 0.10 & 0.25 & 0.50 & 0.75 & 0.90 & \\ \hline
Day (8am-5pm \& 7-11pm)  &0.072&0.086&0.089&0.192&0.392&0.198\\
&(0.006)&(0.005)&(0.004)&(0.005)&(0.009)&(0.002) \\
Peak (5-7pm) &0.326&0.247&0.211&0.371&0.568&0.393\\
&(0.008)&(0.007)&(0.006)&(0.008)&(0.011)&(0.004) \\
Day (8am-5pm \& 7-11pm) $\times$ Treatment &-0.028&-0.005&-0.019&-0.011&0.002&-0.021\\
&(0.011)&(0.009)&(0.006)&(0.009)&(0.015)&(0.004) \\
Peak (5-7pm) $\times$ Treatment &-0.035&-0.010&-0.024&-0.062&-0.100&-0.055\\
&(0.016)&(0.012)&(0.010)&(0.015)&(0.021)&(0.008) \\  \hline
Control variables &Yes&Yes&Yes&Yes&Yes&Yes\\
Household fixed effects &Yes&Yes&Yes&Yes&Yes&Yes\\
Selected partition size, $\hat{l}$ & 5 & 5 & 5 & 5 & 6 & - \\ 
$T$                          & 2160   & 2160  & 2160 & 2160 & 2160 & 2160  \\
$N \times T$                 & 578880 & 578880 & 578880 & 578880  & 578880 & 578880 \\  \hline
\end{tabular}
\end{center}
\caption{\emph{Fixed effects results for a model of electricity consumption. The first five columns present FE-QR results and the last column presents results from estimating a conditional mean model with fixed effects. Standard errors (in parentheses) are obtained by the proposed PWB method.}}
\label{table1.results}
\end{small}
\end{table}
\end{singlespace}

\subsection{Empirical results}

Table \ref{table1.results} presents fixed effects results for the coefficients $\beta_j$ for $j \in \{1,2,3,4\}$ and standard errors corresponding to the point-estimates in parentheses. The standard errors are obtained using the proposed partitioned wild bootstrap (PWB) procedure. The selected partition size, $\hat{l}$, is constant and equal to 5 across quantiles, with the exception of $\tau=0.9$. The first five columns show quantile regression (FE-QR) results at different quantile levels and the last column presents mean fixed effects regression results. To save space, we do not present results on the control variables included in the regression but the results are available upon request. To examine in more detail the performance of the proposed approach in practice, Figure \ref{fig3} compares confidence intervals obtained by alternative approaches. 

Looking at the mean fixed effects results in the last column, we see, as expected, that electricity consumption increases during the day and peak hours relative to night hours. These increases are significant and large, ranging from $\exp(0.393)-1 = 22\%$ percent during the day to 48 percent during peak hours.  When we compare the mean effect with quantile results shown in the first columns, we observe that these estimated effects vary significantly across quantiles. For instance, if we consider consumption during the day, the effect varies from 7 percent at the 0.1 quantile to 48 percent at the 0.9 quantile. When we focus our attention on the parameters of interest, we find results that are consistent with expectations. Recall that the households in the treatment group were charged a slightly lower price during the day than households in the control group, so it is not entirely surprising to find insignificant results across the conditional distribution. On the other hand, when we consider the effect of the price increase during peak hours, we find significant decreases in electricity consumption in the upper half of the distribution, ranging from a modest 2 percent at the 0.5 quantile to 10 percent at the 0.9 quantile. Consistent with the literature, we find that households reduced their usage and are responsive to temporary price increases. 

\begin{figure}
\begin{center}
\centerline{\includegraphics[width=1\textwidth]{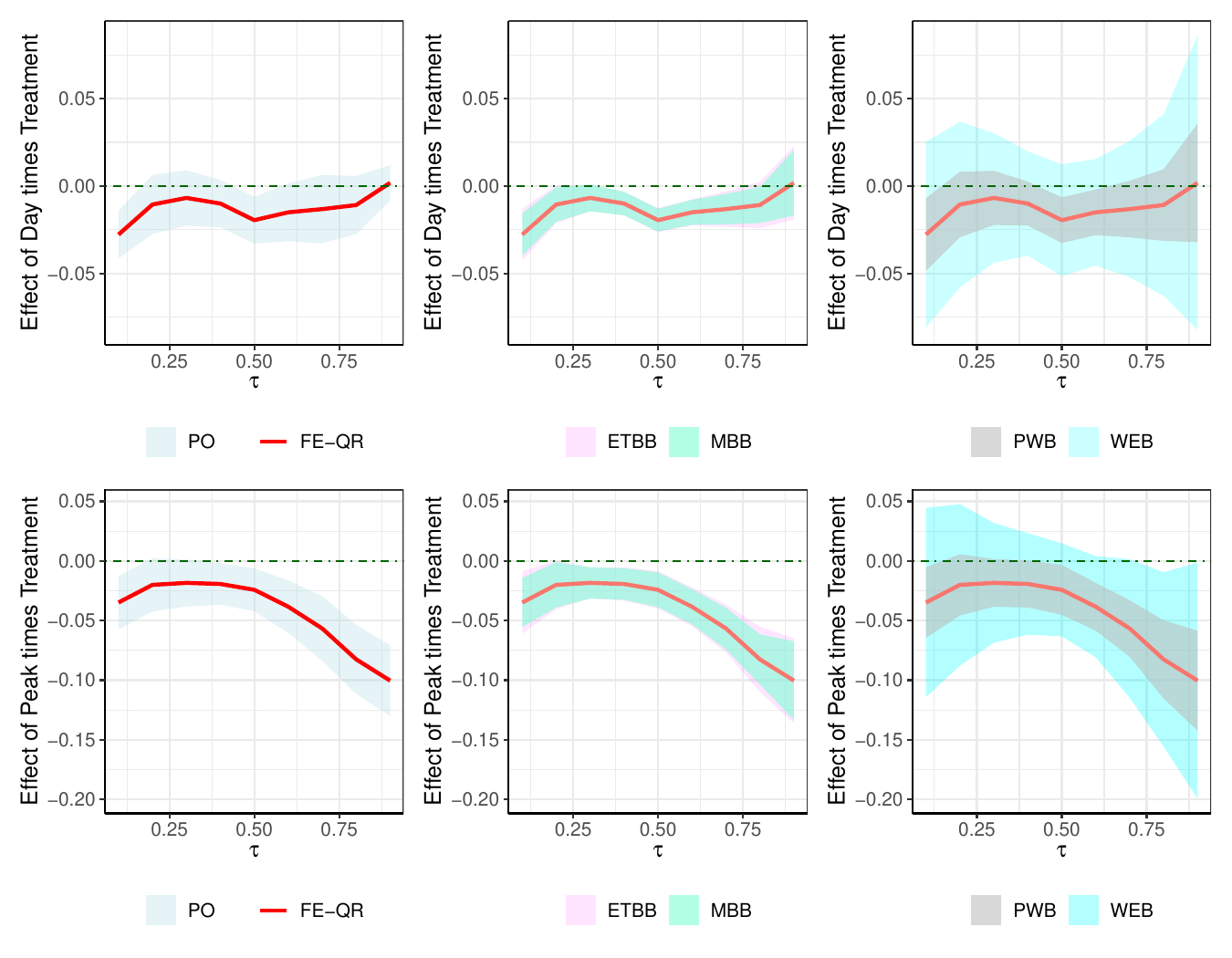}}
\caption{\emph{Time-of-Use effects across quantiles of electricity pricing. FE-QR denotes fixed effects quantile regression and the colored areas show 95\% point-wise confidence intervals. PO is the Powell's Kernel estimator for the asymptotic variance of the FE-QR estimator, MBB denotes moving block bootstrap, ETBB extended tapered block bootstrap, WEB weighted block bootstrap, and PWB partitioned wild bootstrap. }\label{fig3}}
\end{center}
\end{figure}

Figure \ref{fig3} presents fixed effects quantile regression results for $(\beta_3(\tau), \beta_4(\tau))$ for 10 equally spaced $\tau$ in the interval 0.10-0.90. The upper panels present results for the effect of consumption over the day and the lower panels present results for the effect of consumption over the peak hours. In each row, we compare 95 percent pointwise confidence intervals obtained by the proposed PWB procedure and the alternative methods for inference available to practitioners. The left panels show results obtained by estimating the asymptotic covariance matrix under two different assumptions. PO denotes the Kernel estimator considered in Proposition 3.1 in \cite{Kato2012} which assumes i.i.d. data. The next panels show the bootstrap approaches MBB and ETBB applied to panel data.  Finally, the panels in the right show confidence intervals obtained by the proposed PWB estimator and the weighted block bootstrap estimator (WEB) proposed in \cite{GalvaoParkerXiao24}. The WEB method uses one weight per unit, an i.i.d. non-negative random weight with mean and variance both equal to one.

The results in Figure \ref{fig3} highlight the importance of the proposed approach in practice. By ignoring temporal dependence, the kernel estimators produce overly optimistic results, suggesting that the increase in the price during peak hours on electricity usage is negative and significant across the conditional distribution. The application of the bootstrap methods MBB and ETBB lead to a similar observation.  As previously discussed, they most likely underrepresent the size of the confidence intervals as well.  On the other hand, the WEB and PWB methods offer more conservative inference, suggesting that the effects are only significant in the upper tail. Finally, we find significant efficiency improvements from adopting PWB relative to WEB. We believe that this is due to the fact that the treatment varies predictably over time and the absolute price changes are important. 

\section{Conclusion}\label{sec:conclusion}

Statistical inference for quantile regression models of panel data has been a challenge in practice. Existing resampling approaches do not incorporate temporal dependence, which is usually an important concern when practitioners analyze data with a large number of time-series observations. To address this issue, we propose a novel bootstrap method that is easy to implement. The approach re-weights cells of a partition of the estimated residuals to simultaneously satisfy quantile moment conditions conditional on the data and account for how temporal dependence affects standard errors.

We demonstrate that the partitioned wild bootstrap procedure is asymptotically valid for approximating the distribution of the fixed effects estimator. We also investigate the finite sample performance of the approach and find that the new procedure works well. Although we offer an important contribution for panel data observed over many time periods, there are some directions that remain to be investigated. In this work, we focus on the fixed effects estimator, leaving aside the penalized estimator for panel data. We conjecture that the approach can be easily extended to cover the case of penalized estimation of individual effects. We do expect important changes in the consistency result, but we do not offer a formal treatment.  We also note that this bootstrap relies mainly on the existence of separable residuals that can be weighted, and therefore this technique might be applied to other scenarios beyond estimation of a quantile regression model. We leave these investigations to future work.

\newpage

\appendix

\section{Proof of the main results}

The notation $\bm{X}^* \stackrel{p^*}{\longrightarrow} \bm{X}$ denotes convergence in probability of $\bm{X}^*$ to $\bm{X}$ under the resampling distribution, conditional on the observed sample $\bm{S}$.  Similarly, $\exs{\cdot} = \ex{ \cdot | \bm{S}}$ and $\mathrm{P}^*\{\cdot\} = P\{\cdot | \bm{S}\}$ denote the expected value operator and probability calculated conditional on the data, and stochastic order symbols $\Op{\cdot}$ and $\op{\cdot}$ are the usual ones while $\Ops{\cdot}$ and $\ops{\cdot}$ are under the bootstrap distribution conditional on the observed sample. 

Throughout the appendix, $\rho_{\tau}(u) = u(\tau - I(u < 0))$ is the quantile regression loss function and $\psi_\tau(u) = \tau - I(u < 0)$ is the associated score function. We use the term Knight's (1998)\nocite{kK98} identity to refer to the equality $\rho_{\tau}(u-v) - \rho_{\tau}(u) = - v \psi_{\tau}(u) + \int_{0}^{v} (I(u \leq s) - I(u \leq 0)) \ud s$. We define $\bm{\theta}_i = (\bm{\beta}', \alpha_i)'$ for each $i$, $\bm{\alpha} = (\alpha_1, \ldots, \alpha_N)$ and $\bm{\theta} = (\bm{\beta}', \bm{\alpha}')'$ and we suppress their dependency on $\tau$ for notational simplicity.  Let $\bm{z}_{it} = (\bm{x}_{it}', 1)' \in \R^{p+1}$ to conform with $\bm{\theta}_i$.  On the other hand, we maintain the notation $\tilde{\bm{x}}_{it} = \bm{x}_{it} - \varphi_i^{-1} \bm{g}_i$ and its sample counterpart $\check{\bm{x}}_{it} = \bm{x}_{it} - \bar{\varphi}_i^{-1} \bar{\bm{g}}_i$.

Finally, we will use empirical process results repeatedly from~\citet{Kato2012} and~\citet{GalvaoGuVolgushev20}.  References to these articles will be labeled `KGMR' and `GGV' respectively in the text below.  We use $C$ to refer to a constant that does not depend on the parameters or data, and may not be the same constant from line to line.

\begin{proof}[Proof of Lemma~\ref{lem:l_exists}]
Define $\gamma_i^X(k) = \ex{\tilde{\bm{x}}_{it} \tilde{\bm{x}}_{it+k}'}$, $\pi_i(k) = \prob{u_{it} < 0, u_{it+k} < 0 | \bm{x}_{it}, \bm{x}_{it+k}}$ and note that we can write $\gamma_i(k) = \ex{ (\pi_i(k) - \tau^2) \tilde{\bm{x}}_{it} \tilde{\bm{x}}_{it+k}'}$.  These are well-defined and finite for all $i$ and all lags $k = 1, \ldots, T-1$. Fix $N, T$ and suppose they are large enough that $\bm{V}_{NT} - \bm{V}_{NT}^0$ is positive definite.  Then define
\begin{equation*}
  \Gamma_N = \frac{1}{N} \sum_{i=1}^N \begin{bmatrix} \gamma_i(0) & \gamma_i(1) & \ldots & \gamma_i(T-1) \\ \gamma_i(1) & \gamma_i(0) & \ldots & \vdots \\ \vdots & {} & \ddots & {} \\ \gamma_i(T-1) & \ldots & {} & \gamma_i(0) \end{bmatrix}
\end{equation*}
and
\begin{equation*}
  \Gamma^X_N = \frac{1}{N} \sum_{i=1}^N \begin{bmatrix} \gamma_i^X(0) & \gamma_i^X(1) & \ldots & \gamma_i^X(T-1) \\ \gamma_i^X(1) & \gamma_i^X(0) & \ldots & \vdots \\ \vdots & {} & \ddots & {} \\ \gamma_i^X(T-1) & \ldots & {} & \gamma_i^X(0) \end{bmatrix}.
\end{equation*}

This makes it convenient to re-express the difference $\mathcal{V}_N(l) - (\bm{V}_{NT} - \bm{V}_{NT}^0)$.  Let $e_k$ be the $k$-th standard basis vector in $\R^T$ and define $E_k = e_k \otimes I_p$.  Note that $\bm{V}_{NT}^0 = \frac{1}{N} \sum_i \gamma_i(0) = \frac{\tau(1-\tau)}{N} \sum_i \gamma_i^X(0)$.  Write, for $l = 2, \ldots, T$,
\begin{align}
  \mathcal{V}_N(l) - (\bm{V}_{NT} - \bm{V}_{NT}^0) &= \tau (1 - \tau) \sum_{k=1}^l E_k' \Gamma_N^X \sum_{k=1}^l E_k - \sum_{k=1}^T E_k' \Gamma_N \sum_{k=1}^T E_k \notag \\
  {} &= \sum_{k=1}^l E_k' \left( \tau (1 - \tau) \Gamma_N^X - \Gamma_N \right) \sum_{k=1}^l E_k - \sum_{k=l+1}^T E_k' \Gamma_N \sum_{k=l+1}^T E_k, \label{exists_middle}
\end{align}
with the convention that the final sum is empty when $l = T$.

Note that $\mathcal{V}(1) - (\bm{V}_{NT} - \bm{V}_{NT}^0) \leq \zero$, which is ensured by Assumption~\ref{assump:A6}.  Considering $\mathcal{V}(T) - (\bm{V}_{NT} - \bm{V}_{NT}^0)$, note that for a lag length $k$, the corresponding $p \times p$ block of $\tau(1-\tau) \Gamma_N^X - \Gamma_N$ is
\begin{equation} \label{eq:psd_difference}
    \frac{1}{N} \sum_{i=1}^N \ex{ (\tau - \tau^2) \tilde{\bm{x}}_{it} \tilde{\bm{x}}_{it+k} - (\pi_i(k) - \tau^2) \tilde{\bm{x}}_{it} \tilde{\bm{x}}_{it+k} } = 
    \frac{1}{N} \sum_{i=1}^N \ex{ (\tau - \pi_i(k)) \tilde{\bm{x}}_{it} \tilde{\bm{x}}_{it+k} } \geq \zero,
\end{equation}
where the inequality comes from the second part of Assumption~\ref{assump:A6}. Because $\ex{\tilde{\bm{x}}_{it} \tilde{\bm{x}}_{it+k}'}$ is finite and mixing by Assumption~\ref{assume:stationary}, the individual summands decrease to zero as $T$ increases and $\lim_{N,T} \sum_{k=1}^T E_k' \tau(1-\tau) \Gamma_N^X \sum_{k=1}^T E_k$ is finite.  Similarly, $\lim_{N,T} \sum_{k=1}^T E_k' \Gamma_N \sum_{k=1}^T E_k$ is a finite matrix.  Furthermore, as long as there is no perfect dependence between the variables, which would violate the nondegeneracy of $\bm{V}$ assumed in Assumption~\ref{assume:Avar_dependent}, we have
\begin{equation*}
   \lim_{N, T} \sum_{k=1}^T E_k' \tau(1-\tau) \Gamma_N^X \sum_{k=1}^T E_k > \lim_{N,T} \sum_{k=1}^T E_k' \Gamma_N \sum_{k=1}^T E_k.
\end{equation*}
Finally, Assumption~\ref{assume:stationary} implies that the tail of the sequence of covariance matrices contributes a negligible amount to the sum, so there must exist some $L < \infty$ such that we can replace the sum of $T$ terms on the left-hand side of the above display with a finite sum of $L$ terms, that is, such that
\begin{equation*}
   \lim_N \sum_{k=1}^L E_k' \tau(1-\tau) \Gamma_N^X \sum_{k=1}^L E_k > \lim_{N,T} \sum_{k=1}^T E_k' \Gamma_N \sum_{k=1}^T E_k.
\end{equation*}
\end{proof}

\begin{proof}[Proof of Lemma~\ref{lem:lhat}]
As $N,T \rightarrow \infty$, the sample analogs $\bar{\varphi}_i$ and $\bar{\bm{g}}_i$ converge in probability to their population counterparts, implying that $\check{\bm{x}}_{it}$ converges to $\tilde{\bm{x}}_{it}$ in probability.  Also as $b \rightarrow \infty$, each matrix on the left-hand side of~\eqref{eq:lhat_choose} converges in probability to $\ex{\tilde{\bm{x}}_{it} \tilde{\bm{x}}_{it+k}'}$.  Similarly, \citet{GalvaoYoon24} show that the right-hand side converges in probability to one half of $\lim_{N,T \rightarrow \infty} (\bm{V}_{NT} - \bm{V}_{NT}^0)$.  As a result, the left-hand side of~\eqref{eq:lhat_choose} converges in probability to one half of the $\mathcal{V}$ of Lemma~\ref{lem:l_exists}.  Display~\eqref{eq:lhat_choose} implicitly defines an estimating function, and Lemma~\ref{lem:l_exists} implies that it takes a finite number of values over $\{1, 2, \ldots, L\}$, further implying that this equation converges uniformly to a population counterpart that is defined implicitly in Lemma~\ref{lem:l_exists}.  Therefore the $\hat{l}$ of the sample equation that sets it closest to zero must converge in probability to one such $l^{\text{o}}$ of the population equation.
\end{proof}

\begin{proof}[Proof of Theorem~\ref{thm:boot}]
$\bm{\theta}^\ast$ is the minimizer of 
  \begin{equation*}
    \MM_{NT}^\ast(\bm{\theta}) = \frac{1}{N} \sum_{i=1}^N \MM_{i}^\ast(\bm{\theta}_i) = \frac{1}{N} \sum_{i=1}^N \left[ \frac{1}{b l} \sum_{j=1}^b \sum_{s=1}^l \rho_{\tau} \left( y_{ijs}^\ast - \bm{x}_{ijs}^\prime \bm{\beta} - \alpha_i \right) \right],
  \end{equation*}
	where $y_{ijs}^\ast = \bm{x}_{ijs}' \hat{\bm{\beta}} + \hat{\alpha}_{i} + w_{ij} | \hat{u}_{ijs} |$. Define $\Delta_{i}^*(\bm{\theta}_i) = \mathbb{M}_{i}^*(\bm{\theta}_i) - \mathbb{M}_{i}^*(\bm{\theta}_{i0})$.  
 
    Let $\bm{\theta}^\circ$ be an infeasible estimator, defined as the minimizer of 
  \begin{equation*}
    \MM_{NT}^\circ(\bm{\theta}) = \frac{1}{N} \sum_{i=1}^N \MM_{i}^\circ(\bm{\theta}_i) = \frac{1}{N} \sum_{i=1}^N \left[ \frac{1}{b l} \sum_{j=1}^b \sum_{s=1}^l \rho_{\tau} \left( y_{ijs}^\circ - \bm{x}_{ijs}^\prime \bm{\beta} - \alpha_i \right) \right],
  \end{equation*}
	where $y_{ijs}^\circ = \bm{x}_{ijs}' \bm{\beta}_0 + \alpha_{i0} + w_{ij} | u_{ijs} |$.  Similarly define $\Delta_{i}^\circ(\bm{\theta}_i) = \MM_{i}^\circ(\bm{\theta}_i) - \MM_{i}^\circ(\bm{\theta}_{i0})$. This function is convex and is minimized at $\bm{\theta}_{i0}$.

Under Assumptions \ref{b:Gquan}-\ref{b:Gint}, for each $1 \leq i \leq N$,
  \begin{multline}
    \sup_{\bm{\theta}_i \in \RR^{p+1}} |\Delta_{i}^*(\bm{\theta}_i) - \Delta_{i}^\circ(\bm{\theta}_i)| = \Bigg| \frac{1}{bl} \sum_{j=1}^b \sum_{s=1}^l \Big( \rho_\tau \left( w_{ij} |\hat{u}_{ijs}| - \bm{x}_{ijs}'(\bm{\beta} - \hat{\bm{\beta}}) - (\alpha_i - \hat{\alpha}_i) \right) \\
    - \rho_\tau \left( w_{ij} |u_{ijs}| - \bm{x}_{ijs}'(\bm{\beta} - \bm{\beta}_0) - (\alpha_i - \alpha_{i0}) \right) \\
    - \left( \rho_\tau(w_{ij} |\hat{u}_{ijs}|) - \rho_\tau(w_{ij} |u_{ijs}|) \right) \Big) \Bigg| \\
    \leq M \left( 1 + \frac{2}{b} \sum_{j=1}^b |w_{ij}| \right) \left\| \hat{\bm{\beta}} - \bm{\beta}_0 \right\| + \left( 1 + \frac{2}{b} \sum_{j=1}^b |w_{ij}| \right) \left| \hat{\alpha}_i - \alpha_{i0} \right|.
  \end{multline}

  Using the consistency of $\hat{\bm{\theta}}_i$ given by Lemma \ref{lem:feqr}, the average of these differences over $1 \leq i \leq N$ is $\ops{1}$ as $N,T \rightarrow \infty$ and also
  \begin{equation} \label{deltas_close}
    \sup_{\bm{\theta}_i \in \RR^{p+1}} \left| \left[ \Delta_{i}^*(\bm{\theta}_i) - \exs{\Delta_{i}^*(\bm{\theta}_i) } \right] - \left[ \Delta_{i}^\circ(\bm{\theta}_i) - \exs{ \Delta_{i}^\circ(\bm{\theta}_i) } \right] \right| = \ops{1}.
  \end{equation}
  Therefore if we can show consistency for the minimizer $\bm{\theta}^\circ$, it will imply that of $\bm{\theta}^\ast$.

  Define $\mathcal{B}_i(\phi) := \{ \bm{\theta}_i : \| \bm{\theta}_i - \bm{\theta}_{i0} \|_1 \leq \phi \}$, a $\phi$-ball around $\bm{\theta}_{i0}$ for some $\phi > 0$.  We next verify that $\Delta_{i}^\circ(\bar{\bm{\theta}}_i) - \exs{\Delta_{i}^\circ(\bar{\bm{\theta}}_i)}$ is small in this neighborhood. By Knight's identity, $\Delta_{i}^\circ(\bm{\theta}_i) = \VV_{i}^{(1)}(\bm{\theta}_i) + \VV_{i}^{(2)}(\bm{\theta}_i)$, where
  \begin{align}
    \VV_{i}^{(1)}(\bm{\theta}_i) &= -\frac{1}{bl} \sum_{j=1}^b \sum_{s=1}^l  \left\{ \bm{z}_{ijs}' ( \bm{\theta} - \bm{\theta}_{0} ) \right\} \psi_{\tau}( w_{ij} | u_{ijs} | ), \notag \\
    \VV_{i}^{(2)}(\bm{\theta}_i) &= \frac{1}{bl} \sum_{j=1}^b \sum_{s=1}^l \int_{0}^{\bm{z}_{ijs}' (\bm{\theta} - \bm{\theta}_{0})} \left( I( w_{ij} | u_{ijs} | \leq v ) - I( w_{ij} | u_{ijs} | \leq 0 ) \right) \ud v. \label{eq:v2_def}
  \end{align}
Under Assumption \ref{b:Gquan}, $\exs{\VV_{i}^{(1)}(\bm{\theta}_i)} = 0$, and Lemma~\ref{lem:boot_cons1} implies that $\exs{ \VV_{i}^{(2)}(\bm{\theta}_i) } \geq 0$ for each neighborhood $\mathcal{B}_i(\phi)$ with probability increasing to $1$ as $T \rightarrow \infty$.

The bootstrap weights make the bootstrap data $\beta$-mixing, but the mixing coefficients do not satisfy a geometric bound for all lags because the same bootstrap weight is applied to a cell of $l$ residuals at a time.  The events $A = \{y_{is}^* \geq x_{is}'\beta + \alpha_i\} = \{w_{is} > 0\}$ and $B = \{y_{it}^* \geq x_{it}'\beta + \alpha_i\} = \{w_{it} > 0\}$ satisfy $\probs{A, B} = 1$ when $s$ and $t$ are in the same cell but $\probs{A, B} = \probs{A} \probs{B}$ when $s$ and $t$ are members of different cells.  Therefore the mixing coefficients are $\beta(k) = 1$ for $k = 1, \ldots, l-1$ and $\beta(k) = 0$ for $k \geq l$.  That implies that the argument used in the proof of KGMR Theorem 5.1 applies here, but with the bound (assuming that lag length $l$ is small relative to $T$)
	\begin{equation} \label{delcirc_bound}
  \max_{1 \leq i \leq N} \mathrm{P}^* \bigg\{ \sup_{\bm{\theta}_i \in \mathcal{B}_i(\phi)}  \Big| \Delta_{i}^\circ(\bm{\theta}_i) - \exs{\Delta_{i}^\circ(\bm{\theta}_i)} \Big| > \epsilon_\phi \bigg\} \leq C N^{-2}.
\end{equation}
Then, as in the proofs of KGMR Theorems 5.1 and 3.1, $\bm{\theta}^\circ \stackrel{p^*}{\longrightarrow} \bm{\theta}_0$ and also $\bm{\theta}^* \stackrel{p^*}{\longrightarrow} \bm{\theta}_0$.

  We now consider the weak convergence of the estimator. Define the $i$-th contribution to the scores for $\MM_{NT}^*$ with respect to $\bm{\beta}$ and $\alpha_i$ by
  \begin{equation}
    \mathbb{H}_{i}^{\beta\ast}(\bm{\theta}_i) = \frac{1}{b l } \sum_{j=1}^b \sum_{s=1}^l \bm{x}_{ijs} \psi_{\tau}(u^\ast_{ijs} - \bm{x}_{ijs}' ( \bm{\beta} - \hat{\bm{\beta}} ) - ( \alpha_i - \hat{\alpha}_i ))
  \end{equation}
  and
  \begin{equation} \label{score_alpha}
    \mathbb{H}_{i}^{\alpha\ast}(\bm{\theta}_i) = \frac{1}{b l } \sum_{j=1}^b \sum_{s=1}^l  \psi_{\tau}(u^\ast_{ijs} - \bm{x}_{ijs}' ( \bm{\beta} - \hat{\bm{\beta}} ) - ( \alpha_i - \hat{\alpha}_i ) ),
  \end{equation}
  where $u_{ijs}^\ast = w_{ij} | \hat{u}_{ijs} |$. We write
  \begin{multline} \label{alstar_exp}
    \mathbb{H}_{i}^{\alpha\ast}(\bm{\theta}^*_i) = \mathbb{H}_{i}^{\alpha\ast}(\hat{\bm{\theta}}_i) + \left( \mathbb{H}_{i}^{\alpha\ast}(\bm{\theta}^*_i) - \mathbb{H}_{i}^{\alpha\ast}(\hat{\bm{\theta}}_i) - \exs{\mathbb{H}_{i}^{\alpha\ast}(\bm{\theta}^*_i) - \mathbb{H}_{i}^{\alpha\ast}(\hat{\bm{\theta}}_i)} \right) \\
    + \exs{\mathbb{H}_{i}^{\alpha\ast}(\bm{\theta}^*_i) - \mathbb{H}_{i}^{\alpha\ast}(\hat{\bm{\theta}}_i)} 
  \end{multline}

  Recall that $T = b l$, and letting $f_{ijs} := f_i(0 | \bm{x}_{ijs})$, we rewrite $\bar{\varphi}_i = (b l)^{-1} \sum_{j} \sum_{s} f_{ijs}$, $\bar{\bm{g}}_i = (b l)^{-1} \sum_{j} \sum_{s} f_{ijs} \bm{x}_{ijs}$, $\bar{\bm{J}}_i = (b l)^{-1} \sum_{j} \sum_{s} f_{ijs} \bm{x}_{ijs} \bm{x}_{ijs}'$ and $\bar{\bm{D}}_N = N^{-1} \sum_i (\bar{\bm{J}}_i - \bar{\varphi}_i^{-1} \bar{\bm{g}}_i \bar{\bm{g}}_i')$.

  Lemma~\ref{lem:boot_cons2}, the consistency of $\hat{\bm{\theta}}$ and $\bm{\theta}^*_i \stackrel{p^*}{\longrightarrow} \hat{\bm{\theta}}_i$ imply that for all $i$,
  \begin{multline} \label{Hstar_alpha_exp}
    \exs{\mathbb{H}_{i}^{\alpha\ast}(\bm{\theta}^*_i) - \mathbb{H}_{i}^{\alpha\ast}(\hat{\bm{\theta}}_i)} = \bar{\bm{g}}_i'(\bm{\beta}^\ast - \hat{\bm{\beta}}) + \bar{\varphi}_i (\alpha_i^\ast - \hat{\alpha}_i) \\
    + \Ops{ \|\bm{\theta}^*_i - \hat{\bm{\theta}}_i\|^2 } + \Op{ \|\hat{\bm{\theta}}_i - \bm{\theta}_{i0}\|^2 }.
  \end{multline}
  Rewrite~\eqref{alstar_exp} using the definitions and equation~\eqref{Hstar_alpha_exp} as
  \begin{multline} \label{bahadur_alstar}
    \alpha_i^* - \hat{\alpha}_i = -\bar{\varphi}_i^{-1} \bar{\bm{g}}_i' (\bm{\beta}^* - \hat{\bm{\beta}}) - \bar{\varphi}_i^{-1}  \mathbb{H}_{i}^{\alpha\ast}(\hat{\bm{\theta}}_i) - \bar{\varphi}_i^{-1} \left( \mathbb{H}_{i}^{\alpha\ast}(\bm{\theta}^*_i) - \mathbb{H}_{i}^{\alpha\ast}(\hat{\bm{\theta}}_i) - \exs{\mathbb{H}_{i}^{\alpha\ast}(\bm{\theta}^*_i) - \mathbb{H}_{i}^{\alpha\ast}(\hat{\bm{\theta}}_i)} \right) \\
    - \bar{\varphi}_i^{-1} \mathbb{H}_{i}^{\alpha\ast}(\bm{\theta}^*_i) + \Ops{ \| \bm{\theta}^*_i - \hat{\bm{\theta}}_i \|^2 } + \Op{ \| \hat{\bm{\theta}}_i - \bm{\theta}_{i0} \|^2 }.
  \end{multline}

  Similarly,
  \begin{multline} \label{bestar_exp}
    \mathbb{H}_{i}^{\beta\ast}(\bm{\theta}^*_i) = \mathbb{H}_{i}^{\beta\ast}(\hat{\bm{\theta}}_i) + \left( \mathbb{H}_{i}^{\beta\ast}(\bm{\theta}^*_i) - \mathbb{H}_{i}^{\beta\ast}(\hat{\bm{\theta}}_i) - \exs{\mathbb{H}_{i}^{\beta\ast}(\bm{\theta}^*_i) - \mathbb{H}_{i}^{\beta\ast}(\hat{\bm{\theta}}_i)} \right) \\
    + \exs{\mathbb{H}_{i}^{\beta\ast}(\bm{\theta}^*_i) - \mathbb{H}_{i}^{\beta\ast}(\hat{\bm{\theta}}_i)}
  \end{multline}
  and Lemma~\ref{lem:boot_cons2} implies
  \begin{multline} \label{Hstar_beta_exp}
    \exs{\mathbb{H}_{i}^{\beta\ast}(\bm{\theta}^*_i) - \mathbb{H}_{i}^{\beta\ast}(\hat{\bm{\theta}}_i)} = \bar{\bm{J}}_i (\bm{\beta}^* - \hat{\bm{\beta}}) + \bar{\bm{g}}_i (\alpha_i^* - \hat{\alpha}_i) \\
    + \Ops{ \| \bm{\theta}^* - \hat{\bm{\theta}} \|^2 } + \Op{ \| \hat{\bm{\theta}} - \bm{\theta}_0 \|^2 }.
  \end{multline}

  Defining $\mathbb{K}_{i}^{\theta*}(\bm{\theta}_i) = \mathbb{H}_{i}^{\beta\ast}(\bm{\theta}_i) - \bar{\varphi}_i^{-1} \bar{\bm{g}}_i \mathbb{H}_{i}^{\alpha\ast}(\bm{\theta}_i)$ equation~\eqref{bestar_exp} can be combined with equations~\eqref{bahadur_alstar} and~\eqref{Hstar_beta_exp} to write
  \begin{multline}
    \left( \bar{\bm{J}}_i - \bar{\varphi}_i^{-1} \bar{\bm{g}}_i \bar{\bm{g}}_i' \right) ( \bm{\beta}^* - \hat{\bm{\beta}} ) = \mathbb{K}_{i}^{\theta*}(\hat{\bm{\theta}}_i) + \left( \mathbb{K}_{i}^{\theta*}(\bm{\theta}^*_i) - \mathbb{K}_{i}^{\theta*}(\hat{\bm{\theta}}_i) - \exs{\mathbb{K}_{i}^{\theta*}(\bm{\theta}^*_i) - \mathbb{K}_{i}^{\theta*}(\hat{\bm{\theta}}_i)} \right) \\
    + \mathbb{K}_{i}^{\theta*}(\bm{\theta}^*_i) + \Ops{ \| \bm{\theta}^* - \hat{\bm{\theta}} \|^2 } + \Op{ \| \hat{\bm{\theta}} - \bm{\theta}_0 \|^2 }.
  \end{multline}

  Average over $i$ to find
  \begin{multline} \label{betastar_expansion}
    \bar{\bm{D}}_N \left( \bm{\beta}^* - \hat{\bm{\beta}} \right) + \ops{\|\bm{\beta}^\ast - \hat{\bm{\beta}}\|} = \frac{1}{N} \sum_{i=1}^N \mathbb{K}_{i}^{\theta*}(\hat{\bm{\theta}}_i) \\
    + \frac{1}{N} \sum_{i=1}^N \left( \mathbb{K}_{i}^{\theta*}(\bm{\theta}_{i}^*) - \mathbb{K}_{i}^{\theta*}(\hat{\bm{\theta}}_i) - \exs{\mathbb{K}_{i}^{\theta*}(\bm{\theta}^*_i) + \mathbb{K}_{i}^{\theta*}(\hat{\bm{\theta}}_i)} \right) + \frac{1}{N} \sum_{i=1}^N \mathbb{K}_{i}^{\theta*}(\bm{\theta}_i^\ast) \\
    + \Ops{ \sup_i (\alpha_i^* - \hat{\alpha}_i)^2 } + \Op{ \sup_i \|\hat{\bm{\theta}}_i - \bm{\theta}_{i0}\|^2 }.
  \end{multline}

  Next we determine the stochastic order of the terms in~\eqref{betastar_expansion}.  First, $\bar{\bm{D}}_N$ is the average of bounded terms under Assumptions~\ref{a:minf} and~\ref{a:boundedX}, so given Assumptions~\ref{assume:stationary} and \ref{assume:Avar_dependent}, it converges in probability to $\bm{D}$ and is nonsingular with probability increasing to one.  Second, because $\bm{\theta}^\ast$ is the minimizer of the quantile regression objective function, each $\mathbb{K}_i^{\theta\ast}(\bm{\theta}_i^\ast) = \Ops{T^{-1}}$ and therefore their average is as well. 
  Third, Step 3 in the proof of KGMR Theorems 3.2 and 5.1 implies that $\sup_i \|\hat{\bm{\theta}}_i - \bm{\theta}_{i0} \| = \Op{T^{-1} \log T}$.  Lemma~\ref{lem:alphasmall} shows that $\Ops{ \sup_i (\alpha_i^* - \hat{\alpha}_i)^2 } = \Ops{T^{-1}\log T}$.  Lemma~\ref{lem:4K} shows that
  \begin{multline}
    \frac{1}{N} \sum_{i=1}^N \left( \mathbb{K}_{i}^{\theta*}(\bm{\theta}_{i}^*) - \mathbb{K}_{i}^{\theta*}(\hat{\bm{\theta}}_i) - \exs{\mathbb{K}_{i}^{\theta*}(\bm{\theta}^*_i) + \mathbb{K}_{i}^{\theta*}(\hat{\bm{\theta}}_i)} \right) \\
    {} = \Ops{ T^{-1/2} \log T \|\bm{\beta}^\ast - \hat{\bm{\beta}}\|^{1/2} + T^{-1} (\log T)^2 + N^{-1/2} T^{-2/3}}. \label{4Kstatement}
  \end{multline}
  The first term on the right-hand side of~\eqref{betastar_expansion} is $\Ops{(NT)^{-1/2}}$.  Putting these facts together we find
  \begin{multline*}
    \left\| \bm{\beta}^* - \hat{\bm{\beta}} \right\| = \Ops{(NT)^{-1/2}} \\
    + \Ops{ T^{-1/2} \log T \|\bm{\beta}^\ast - \hat{\bm{\beta}}\|^{1/2} + T^{-1} (\log T)^2 + N^{-1/2} T^{-2/3}}.
  \end{multline*}
  Using the fact that $\delta \leq a + b\delta^{1/2} \Rightarrow \delta \leq 4\max\{a, b^2\}$ and the rate restriction given in the statement of the theorem, 
  \begin{equation*}
    \left\| \bm{\beta}^* - \hat{\bm{\beta}} \right\| = \Ops{ T^{-1} (\log T)^2 + (NT)^{-1/2} } = \Ops{ (NT)^{-1/2} }.
  \end{equation*}

Then under Assumption~\ref{assump:A6} and assuming that $l$ is chosen to satisfy~\eqref{lstar_choose}, the finite third moments (implied by Assumption~\ref{a:boundedX}, Lemma~\ref{lem:lstar_choose} below and arguing as in the proof of KGMR Theorem 5.1),
\begin{equation*}
    \sqrt{NT} \frac{1}{N} \sum_{i=1}^N \mathbb{K}_i^{\theta\ast}(\hat{\bm{\theta}}_i) \dconvs \calN\left( \zero_p, 
\bm{V} \right),
\end{equation*}
where $\bm{V}$ is defined in Assumption~\ref{assume:Avar_dependent} \citep[Theorem 5.20]{White01}.
\end{proof}

\begin{lemma} \label{lem:lstar_choose}
Suppose that Assumptions \ref{assume:stationary}--\ref{assump:A6}, \ref{b:wiid}--\ref{b:Gint} are satisfied and assume that $N (\log T)^4 / T = o(1)$.  By choosing $\hat{l}$ as described in~\eqref{eq:lhat_choose}, the bootstrap variance estimator is a consistent estimate of $\bm{V}$ defined in Assumption~\ref{assume:Avar_dependent}.
\end{lemma}
\begin{proof}[Proof of Lemma~\ref{lem:lstar_choose}]

The first term in \eqref{betastar_expansion} satisfies
  \begin{align*}
    \exs{\mathbb{K}_{i}^{\theta*}(\hat{\bm{\theta}}_i)} & = \exs{\frac{1}{bl} \sum_{j=1}^b \sum_{s=1}^{l}  \check{\bm{x}}_{ijs} \psi_{\tau} \left( u_{ijs}^\ast  \right) } \\
  & = \frac{1}{bl} \sum_{j=1}^b \sum_{s=1}^{l} \check{\bm{x}}_{ijs} \exs{\psi _{\tau} \left( w_{ij}  \right)} = \zero_p. 
   \end{align*}

Next we calculate the variance of this term under the bootstrap measure.  For each $i$, since this expectation is conditional on the data and given Assumptions~\ref{b:wiid} and~\ref{b:Gquan} on the weights,
\begin{align}
    bl \exs{\mathbb{K}_{i}^{\theta*}(\hat{\bm{\theta}}_i) \mathbb{K}_{i}^{\theta*}(\hat{\bm{\theta}}_i)'} & = \frac{1}{bl} \exs{ \sum_{j=1}^b \sum_{s=1}^{l}  \check{\bm{x}}_{ijs} \psi_{\tau} \left( u_{ijs}^\ast  \right) \sum_{j=1}^b \sum_{s=1}^{l}  \check{\bm{x}}_{ijs}' \psi_{\tau} \left( u_{ijs}^\ast  \right) } \notag \\
    {} &= \frac{1}{bl} \sum_{j=1}^b \left( \sum_{s=1}^{l} \check{\bm{x}}_{ijs} \right) \left( \sum_{s=1}^{l} \check{\bm{x}}_{ijs} \right)' \exs{ \psi _{\tau}^2 \left( w_{ij}  \right) } \notag \\
    {} &= \frac{\tau(1-\tau)}{bl} \sum_{j=1}^b \left( \sum_{s=1}^l \check{\bm{x}}_{ijs} \check{\bm{x}}_{ijs}' + 2 \sum_{k=1}^{l-1} \sum_{s=1}^{l-k}  \check{\bm{x}}_{ijs} \check{\bm{x}}_{ij(s+k)}' \right). \label{bvar}
\end{align}

If $\hat{l} \pconv l^{\text{o}} < \infty$ and $b \rightarrow \infty$ (recall we assume $T = bl$), under consistency of sample estimators $\bar{\varphi}_i \pconv \varphi$ and $\bar{\bm{g}}_i \pconv \bm{g}_i$ for all $i$, we have for each $i$
\begin{align}
    bl \exs{\mathbb{K}_{i}^{\theta*}(\hat{\bm{\theta}}_i) \mathbb{K}_{i}^{\theta*}(\hat{\bm{\theta}}_i)'} &\pconv \tau(1-\tau) \ex{\frac{1}{l} \sum_{t=1}^{l-1} \tilde{\bm{x}}_{it} \sum_{t=1}^{l-1} \tilde{\bm{x}}_{it}' } \notag \\
     {} &= \tau(1-\tau) \ex{\tilde{\bm{x}}_{it} \tilde{\bm{x}}_{it}'}  + 2\tau(1-\tau) \sum_{k=1}^{l^{\text{o}}-1} \frac{l^{\text{o}}-k}{l^{\text{o}}} \ex{\tilde{\bm{x}}_{it} \tilde{\bm{x}}_{it+k}'}. \label{eq:Avar_boot}
\end{align}
When averaging over $i$, we find that this matches the $\mathcal{V}(l^{\text{o}})$ described in Lemma~\ref{lem:l_exists}, indicating that it converges in probability (conditional on the data) to $\bm{V}$ as $N, T \rightarrow \infty$.
\end{proof}

\begin{lem}[Preliminary size estimates] \label{lem:alphasmall}
Under the assumptions in Theorem~\ref{thm:boot},
\begin{align*}
    \sup_i |\alpha_i^\ast - \hat{\alpha}_i| &= \Ops{T^{-1/2} (\log T)^{1/2}} \\
    \|\bm{\beta}^\ast - \hat{\bm{\beta}}\| &= \ops{T^{-1/2} (\log T)^{1/2}}.
\end{align*}
\end{lem}

\begin{proof}[Proof of Lemma~\ref{lem:alphasmall}]
This proof follows Step 3 of the proof of KGMR Theorems 3.2 and 5.1, but it is tailored to this bootstrap distribution.  We will check the stochastic order of the terms on the right-hand side of~\eqref{bahadur_alstar}.

First we make a size estimate of $\|\bm{\beta}^\ast - \hat{\bm{\beta}}\|$.   The first two terms in~\eqref{betastar_expansion} are averages of $NT$ observations, and as such are $\ops{T^{-1/2}}$.  Recall that in Step 3 of the proof of KGMR Theorem 3.2 make stochastic order estimates of $\hat{\bm{\beta}}$ and $\hat{\alpha}_i$.  Then using the other facts listed in the paragraph below the display \eqref{betastar_expansion} can be estimated as
\begin{equation} \label{eq:betaest}
    \|\bm{\beta}^\ast - \hat{\bm{\beta}}\| = \ops{T^{-1/2}} + \Ops{\sup_i |\alpha_i^\ast - \hat{\alpha}_i|^2} + \Op{T^{-1} \log T}
\end{equation}
and we have
\begin{multline} \label{bigalpha}
    \sup_i |\alpha_i^\ast - \hat{\alpha}_i| + \ops{|\alpha_i^\ast - \hat{\alpha}_i|} = \Ops{\sup_i \left| \mathbb{H}_{i}^{\alpha\ast}(\hat{\bm{\theta}}_i) \right|} \\
    + \Ops{\sup_i \left| \mathbb{H}_{i}^{\alpha\ast}(\bm{\theta}^*_i) - \mathbb{H}_{i}^{\alpha\ast}(\hat{\bm{\theta}}_i) - \exs{\mathbb{H}_{i}^{\alpha\ast}(\bm{\theta}^*_i) - \mathbb{H}_{i}^{\alpha\ast}(\hat{\bm{\theta}}_i)} \right|} \\
    + \ops{T^{-1/2}} + \Op{T^{-1} \log T}.
\end{multline}

We will use KGMR Corollary C.1 with the first term of~\eqref{bigalpha}, which requires the computation of the variance $\sigma_q^2(\psi_\tau) = \sum_{t=1}^q \psi(u_t^\ast) / \sqrt{q}$ for $1 \leq q \leq T$.  This is
\begin{align*}
\sigma_q^2(\psi_\tau) &= \exs{\psi_\tau^2(u_t^\ast)} + 2 \sum_{k=1}^{q-1} (1 - j/q) \exs{\psi_\tau(u_t^\ast) \psi_\tau(u_{t+k}^\ast)} \\
    \sigma_q^2(\psi_\tau) &= \tau(1-\tau) + 2 \sum_{j=1}^{q \wedge l -1} (1 - j/q) \tau(1-\tau) \\
    {} &= \tau(1-\tau) \left[ 1 + (q \wedge l - 1) \left( 2 - (q \wedge l)/q \right) \right].
\end{align*}

However, as shown in the main text, $l^\text{o} = \Op{1}$, and we can consider $\sigma_q^2(\psi_\tau)$ bounded by a constant.  Then by KGMR Proposition C.1 and Corollary C.1 (choose $U = 1$ and note that because $\sigma_q^2(\psi_\tau) \leq \tau(1-\tau) q$, the assumptions of Proposition C.1 are satisfied),
\begin{equation*}
\probs{ T \sup_i \left| \mathbb{H}_{i}^{\alpha\ast}(\hat{\bm{\theta}}_i) \right| > C\left( \sqrt{(s \vee 1) T} + sq \right) } \leq 2e^{-s} + 2 \left\lfloor \frac{T}{2q} \right\rfloor \beta(q),
\end{equation*}
where $\beta(q) = 1 - \tau$ for $q = 1, \ldots, l-1$ and $0$ for $q \geq l$.  Taking $s = C_s \log T$ and $q = C_q \log T$ as in GGV Lemma 5, we can ensure that the right-hand side of the above expression is bounded by $T^{-\kappa}$, implying that
\begin{equation} \label{eq:hstar_est}
\sup_i \left| \mathbb{H}_{i}^{\alpha\ast}(\hat{\bm{\theta}}_i) \right| = \Ops{ T^{-1/2} (\log T)^{1/2} }.
\end{equation}

Similarly, we would like to use KGMR Proposition C.2 to show that the final estimate is of smaller order.  To this end we can define a set of functions $\mathcal{G}_\delta = \{g(\theta) = I(u_t < \alpha + \bm{x}\tr\bm{\beta}) - I(u_t < 0), |\alpha| \leq \delta, \|\bm{\beta}\| \leq \delta\}$, and using calculations as in the end of Lemma~\ref{lem:boot_cons2} and KGMR Lemma C.1, we find $\sigma_q^2(\mathcal{G}_\delta) = C \delta^{1/2}$.  Therefore their Proposition C.2 implies that for some constants (assuming $\log \delta \sim \log T$)
\begin{equation*}
\prob{ \left\| \sum_t g(\theta) \right\|_{\mathcal{G}_\delta} \geq C \left( T^{1/2} \delta^{1/2} (\log T)^{1/2} + \delta^{1/2} s^{1/2} T^{1/2} + sq \right) } \leq 2e^{-s} + 2T\beta(q).
\end{equation*}
Take  $s = C_s \log T$ and $q = C_q \log T$ as in GGV Lemma 5, so that the right-hand side is bounded by $T^{-\kappa}$.  Then this implies that
\begin{multline}
    \sup_i \left| \mathbb{H}_{i}^{\alpha\ast}(\bm{\theta}^*_i) - \mathbb{H}_{i}^{\alpha\ast}(\hat{\bm{\theta}}_i) - \exs{\mathbb{H}_{i}^{\alpha\ast}(\bm{\theta}^*_i) - \mathbb{H}_{i}^{\alpha\ast}(\hat{\bm{\theta}}_i)} \right| \\
    {} = \Ops{ T^{-1/2} \sup_i \|\bm{\theta}^\ast_i - \hat{\bm{\theta}}_i\|^{1/2} (\log T)^{1/2} + T^{-1} (\log T)^2} \\
    {} = \ops{T^{-1/2} (\log T)^{1/2}} \label{eq:4hs_est}
\end{multline}
by consistency of $\bm{\theta}^\ast$. Putting estimates~\eqref{eq:hstar_est} and~\eqref{eq:4hs_est} into~\eqref{bigalpha} and comparing the orders of all the terms implies the result for $\sup_i |\alpha_i^\ast - \hat{\alpha}_i|$.  Using this in~\eqref{eq:betaest} implies the result for $\|\bm{\beta}^\ast - \hat{\bm{\beta}}\|$.
\end{proof}

The next two Lemmas are similar to the proof of Theorem 1 in~\citet{FengHeHu11} or that of Lemma S.2 in~\citet{LamarcheParker23}, altered to apply to stationary data.  In the proofs of the main theorems above, individual observations are indexed by unit, cell and time period within cell $(i, j, s)$, which is equivalent to rewriting the index using only individual and time period, that is, $(i, t)$, as is done below (although for some weights indexed by $s \neq t$ we will have $w_{is} = w_{it}$).

\begin{lem} \label{lem:boot_cons1}
    Recall $\bm{z}_{it} = (\bm{x}_{it}', 1)'$ and let $\bm{\Delta} \in \RR^{p+1}$.   Then under Assumptions \ref{assume:stationary}-\ref{a:boundedX} and \ref{b:wiid}-\ref{b:Gint}, as $T \rightarrow \infty$ the following approximation holds for all $i$: as $T \rightarrow \infty$, uniformly for $\bm{\Delta}$ in an $\phi$-neighborhood of the origin,
    \begin{multline*}
        \exs{ \frac{1}{T} \sum_{t=1}^T \int_0^{\bm{z}_{it}'\bm{\Delta}} \left( I(w_{it} |u_{it}| \leq s) - I(w_{it} |u_{it}| \leq 0) \right) \ud s} \\
        = \frac{1}{T} \sum_{t=1}^T f_i(0 | \bm{x}_{it}) \bm{\Delta}' \bm{z}_{it} \bm{z}_{it}'\bm{\Delta} + O(\|\bm{\Delta}\|^2) + \Op{T^{-1/2} (\log T)^{1/2} \phi^{1/2} + T^{-1} (\log T)^2}.
        \end{multline*}
\end{lem}

\begin{proof}[Proof of Lemma~\ref{lem:boot_cons1}]
    Abbreviate the average of the integrals on left-hand side of the above expression as $V_i^\ast(\bm{\Delta})$.  We calculate the first two moments of $\exs{V_i^\ast(\bm{\Delta})}$ conditional on $\bm{X}_i$. We use the identity $I(u \leq v) - I(u \leq 0) = I(0 < u \leq v) - I(v < u \leq 0)$.  Consider a single $(i, t)$ observation to compute the first moment.  Conditioning on $\bm{x}_{it}$,
    \begin{multline}
        \ex{\exs{V_i^\ast(\bm{\Delta})} \big| \bm{x}_{it}} = \ex{\exs{\int_{0}^{\bm{z}_{it}'\bm{\Delta}} I(0 < w_{i} | u_{it} | \leq v ) \ud v } I(\bm{z}_{it}'\bm{\Delta} > 0) \Big| \bm{x}_{it}} \\
        {} - \ex{\exs{\int_0^{\bm{z}_{it}'\bm{\Delta}} I(v < w_{i} | u_{it} | \leq 0 ) \ud v } I(\bm{z}_{it}'\bm{\Delta} \leq 0) \Big| \bm{x}_{it}}. \label{bigexp}
    \end{multline}
    Taking just the first half of the right-hand side, Fubini's theorem and Assumptions~\ref{a:minf}, \ref{b:Gsupp}, \ref{b:Gint} and \ref{a:boundedX} imply that
    \begin{align*}
    \text{E} \Bigg[ \text{E}^\ast \Bigg[ \int_{0}^{\bm{z}_{it}'\bm{\Delta}} I(0 &< w_{i} | u_{it} | \leq v ) \ud v \Bigg] I(\bm{z}_{it}'\bm{\Delta} > 0) \Big| \bm{x}_{it} \Bigg] \\
    {} &= \exs{\int_{0}^{\bm{z}_{it}'\bm{\Delta}} F_i(v / w_{i} | \bm{x}_{it}) ) - F_i(-v / w_{i} | \bm{x}_{it} ) \ud v I(w_{i} > 0) } \\
    {} &= \exs{\int_{0}^{\bm{z}_{it}'\bm{\Delta}} [2 f_i(0 | \bm{x}_{it}) (v / w_{i}) + O(v^2)] \ud v I(w_{i} > 0) } \\
    {} &= f_i(0 | \bm{x}_{it}) (\bm{z}_{it}'\bm{\Delta})^2 \exs{w_{i}^{-1} I(w_{i} > 0)} + O(\|\bm{\Delta}\|^2) \\
    {} &= \frac{1}{2} f_i(0 | \bm{x}_{it}) (\bm{z}_{it}'\bm{\Delta})^2 + O(\|\bm{\Delta}\|^2).
    \end{align*}
    Analogous calculations show that the integral on the event $\{\bm{z}_{it}'\bm{\Delta} \leq 0\}$ in~\eqref{bigexp} is the same, and adding them together and using this for all the observations in the sum results in
    \begin{equation*}
        \ex{\exs{V_i^\ast(\bm{\Delta})} \big| \bm{X}_i} = \frac{1}{T} \sum_{t=1}^T f_i(0 | \bm{x}_{it}) \bm{\Delta}'\bm{z}_{it} \bm{z}_{it}'\bm{\Delta} + O(\|\bm{\Delta}\|^2).
    \end{equation*}

    Next, bound the difference between $\exs{V_i^\ast(\bm{\Delta})}$ and its expected value (conditional on $\bm{X}_i$).  The absolute value of each integrand in $V_i^\ast$ is bounded by one, and under Assumptions~\ref{assume:stationary} and~\ref{a:boundedX} and using KGMR Lemma C.1 (use $\delta = 1$),
    \begin{equation*}
        \var{ \exs{V_i^\ast(\bm{\Delta})} \big| \bm{X}_i } \leq \var{ \frac{1}{T} \sum_{t=1}^T \bm{z}_{it}'\bm{\Delta}} = \Op{T^{-1} \|\bm{\Delta}\|}.
    \end{equation*}
    Then define $\sigma_q^2 = C_\sigma \phi$, and by choosing $s$ and $q$ as in GGV Lemma 5, KGMR Corollary C.1 implies that as $T \rightarrow \infty$,
    \begin{equation*}
        \exs{V_i^\ast(\bm{\Delta})} - \ex{ \exs{V_i^\ast(\bm{\Delta})} | \bm{X}_i} = \Op{T^{-1/2} (\log T)^{1/2} \phi^{1/2} + T^{-1} (\log T)^2}.
    \end{equation*}
\end{proof}

\begin{lem} \label{lem:boot_cons2}
    Recall $\bm{z}_{it} = (\bm{x}_{it}', 1)'$ and let $\bm{\delta}, \bm{\Delta} \in \RR^{p+1}$.  Under Assumptions \ref{assume:stationary}-\ref{a:boundedX} and \ref{b:wiid}-\ref{b:Gint}, for each $i$, as $T \rightarrow \infty$, and uniformly for $\bm{\delta}, \bm{\Delta}$ in a $\phi$-neighborhood of the origin,
    \begin{multline*}
        \exs{ \frac{1}{T} \sum_{t=1}^T \left( \psi_\tau(w_{it} |u_{it} - \bm{z}_{it}'\bm{\Delta}| - \bm{z}_{it}' \bm{\delta}) - \psi_\tau(w_{it} |u_{it} - \bm{z}_{it}'\bm{\Delta}|) \right) } \\
        = \frac{1}{T} \sum_{t=1}^T f_i(0 | \bm{x}_{it}) \bm{z}_{it}'\bm{\delta} + O \left( \|\bm{\Delta}\|^2 \right) + O \left( \|\bm{\delta}\|^2 \right) + \Op{T^{-1/2} (\log T)^{1/2} \phi^{1/2} + T^{-1} (\log T)^2}.
    \end{multline*}
\end{lem}     

\begin{proof}[Proof of Lemma~\ref{lem:boot_cons2}]
Start with the identity $\psi(a - b) - \psi(a) = -I(a < b) + I(a < 0) = I(b \leq a < 0) - I(0 \leq a < b)$.  This implies
\begin{multline*}
\exs{ \psi_\tau(w_{it} |u_{it} - \bm{z}_{it}'\bm{\Delta}| - \bm{z}_{it}' \bm{\delta}) - \psi_\tau(w_{it} |u_{it} - \bm{z}_{it}'\bm{\Delta}|) } \\
= \exs{ I(\bm{z}_{it}'\bm{\delta} \leq w_{it} |u_{it} - \bm{z}_{it}'\bm{\Delta}| < 0) - I(0 \leq w_{it} |u_{it} - \bm{z}_{it}'\bm{\Delta}| < \bm{z}_{it}'\bm{\delta}) }.
\end{multline*}
Consider the expectation of the first term on the right-hand side conditional on $\bm{x}_{it}$, which is nonzero when $\bm{z}_{it}'\bm{\delta} < 0$:
\begin{multline*}
\ex{ \exs{ I(\bm{z}_{it}'\bm{\delta} \leq w_{it} |u_{it} - \bm{z}_{it}'\bm{\Delta}| < 0) } | \bm{x}_{it}} \\
=  \ex{ \exs{ I(\bm{z}_{it}'\bm{\Delta} - \bm{z}_{it}'\bm{\delta} / w_{it} \leq u_{it} \leq \bm{z}_{it}'\bm{\Delta} + \bm{z}_{it}'\bm{\delta} / w_{it}) I(w_{it} < 0)} | \bm{x}_{it}} I(\bm{z}_{it}'\bm{\delta} < 0).
\end{multline*}
Fubini's Theorem, two Taylor expansions and Assumptions~\ref{a:minf}, \ref{a:boundedX}, \ref{b:Gsupp} and~\ref{b:Gint} imply that this can be rewritten
\begin{multline*}
\exs{ \left( F_i(\bm{z}_{it}'\bm{\Delta} + \bm{z}_{it}'\bm{\delta} / w_{it} | \bm{x}_{it}) - F_i(\bm{z}_{it}'\bm{\Delta} - \bm{z}_{it}'\bm{\delta} / w_{it} | \bm{x}_{it}) \right) I(w_{it} < 0) } I(\bm{z}_{it}'\bm{\delta} < 0) \\
= \left( f_i(0 | \bm{x}_{it}) \bm{z}_{it}'\bm{\delta} + O\left( (\|\bm{\Delta}\| + \|\bm{\delta}\|)^2 \right) \right) I(\bm{z}_{it}'\bm{\delta} < 0).
\end{multline*}
Analogous calculations imply an identical approximation on the event $\{\bm{z}_{it}'\bm{\delta} > 0\}$.  Combining these two expressions implies that
\begin{multline*}
\ex{ \exs{ \psi_\tau(w_{it} |u_{it} - \bm{z}_{it}'\bm{\Delta}| - \bm{z}_{it}' \bm{\delta}) - \psi_\tau(w_{it} |u_{it} - \bm{z}_{it}'\bm{\Delta}|) } | \bm{x}_{it}} = f_i(0 | \bm{x}_{it}) \bm{z}_{it}'\bm{\delta} \\
+ O\left( (\|\bm{\Delta}\| + \|\bm{\delta}\|)^2 \right).
\end{multline*}
Because $2ab \leq a^2 + b^2$ we can break the order of the approximation remainder into two pieces as in the statement of the lemma.

Noting that $|\psi_\tau(a - b) - \psi_\tau(a)| \leq I(|a| \leq |b|)$, recalling Assumption~\ref{b:Gsupp} and letting $c_{\text{min}} = c_1 \wedge c_2$, we have
\begin{multline*}
\var{\exs{\frac{1}{T} \sum_{t=1}^T \left( \psi_\tau(w_{it} |u_{it} - \bm{z}_{it}'\bm{\Delta}| - \bm{z}_{it}' \bm{\delta}) - \psi_\tau(w_{it} |u_{it} - \bm{z}_{it}'\bm{\Delta}|) \right)}} \\
\leq \var{\frac{1}{T} \sum_{t=1}^T I(|u_{it} - \bm{z}_{it}'\bm{\Delta}| \leq c_{\text{min}}^{-1} |\bm{z}_{it}'\bm{\delta}|)}.
\end{multline*}
Using KGMR Lemma C.1 (with $\delta = 1$) and Assumption~\ref{assume:stationary}, this latter variance is $O(\|\bm{\delta}\|)$.  Then KGMR Proposition C.2, applied to the class of functions $\mathcal{G}_\phi = \{(w, u, \bm{z}) \mapsto I(w|u - \bm{z}\tr \bm{D}| < \bm{z}\tr \bm{d}) - I(w|u - \bm{z}\tr \bm{D}| < 0), \|\bm{d}\|, \|\bm{D}\| \leq \phi\}$ (with $s = C_s \log T$ and $q = C_q \log T$ chosen appropriately), implies that
\begin{multline*}
    \sup_{g \in \mathcal{G}_\phi} \left\lvert \exs{\frac{1}{T} \sum_t g_t(\bm{d}, \bm{D})} - \ex{\exs{\frac{1}{T} \sum_t g_t(\bm{d}, \bm{D})}} \right\rvert \\ = \Op{T^{-1/2} (\log T)^{1/2} \phi^{1/2} + T^{-1} (\log T)^2}.
\end{multline*}
This implies the approximation holds uniformly in the neighborhood asymptotically.
\end{proof}

Lemma~\ref{lem:4K} shows~\eqref{4Kstatement}, while Lemmas~\ref{lem:alphas}-\ref{lem:4Kexp} below support Lemma~\ref{lem:4K}.  For reference define components of the sum in~\eqref{lem:4K} as
\begin{equation} \label{kcal}
    \mathcal{K}_i(\bm{\theta}_1, \bm{\theta}_2) = \mathbb{K}_{i}^{\theta*}(\bm{\theta}_1) - \mathbb{K}_{i}^{\theta*}(\bm{\theta}_2) - \exs{\mathbb{K}_{i}^{\theta*}(\bm{\theta}_1) + \mathbb{K}_{i}^{\theta*}(\bm{\theta}_2)}.
\end{equation}
To find the stochastic order of these terms it is convenient to define the class of functions
\begin{equation} \label{def:Gdelta}
  \mathcal{G}(\delta) = \{(w, u, \bm{z}) \mapsto (I(w|u| < \bm{z}\tr \bm{b}_1) - I(w|u| < \bm{z}\tr \bm{b}_2) ) I(\|\bm{z}\| \leq M), \|\bm{b}_1 - \bm{b}_2\| \leq \delta\}
\end{equation}
and define
\begin{equation} \label{def:empproc}
  \|\mathbb{P}_{iT}^\ast - \mathbb{P}_{iT}\|_{\mathcal{G}(\delta)} = \sup_{g \in \mathcal{G}(\delta)} \left| \frac{1}{T} \sum_{t=1}^T \left( g(w_{it}, u_{it}, \bm{z}_{it}) - \exs{g(w_{it}, u_{it}, \bm{z}_{it})} \right) \right|.
\end{equation}
Finally, let
\begin{equation} \label{def:chi}
  \chi_T(\delta) = T^{-1/2} \log T \delta^{1/2} + T^{-1} (\log T)^2.
\end{equation}
KGMR Proposition C.2 and the argument in GGV Lemma 5 can be applied to the class $\mathcal{G}(\delta)$ (it is their $\mathcal{G}_2(\delta)$) with the bootstrap weight distribution and conditional on the observations.  GGV Lemma 5 states that there is a constant $C$ such that for any $\kappa > 1$,
\begin{equation} \label{GGVL5}
  \probs{ \sup_{\delta \in (0, 1)} \frac{\|\mathbb{P}_{iT}^\ast - \mathbb{P}_{iT}\|_{\mathcal{G}(\delta)}}{\chi(\delta)} \geq C \kappa^2 } \leq T^{-\kappa}.
\end{equation}
These definitions will be used repeatedly in Lemmas~\ref{lem:4K} and~\ref{lem:alphas}.  It is also useful to define
\begin{equation} \label{def:alphas}
  \alpha_i^{\circ\ast} = \argmin_a \frac{1}{T} \sum_{t=1}^T \rho_\tau(y_{it}^\ast - \bm{x}_{it}' \hat{\bm{\beta}} - a),
\end{equation}
that is, the estimator of the $i$-th individual effect conditional on the data and $\hat{\bm{\beta}}$.  Then define $\bm{\theta}^{\circ\ast} = (\hat{\bm{\beta}}, \alpha_1^{\circ\ast}, \ldots, \alpha_N^{\circ\ast})$ and $\bm{\theta}_i^{\circ\ast} = (\hat{\bm{\beta}}, \alpha_i^{\circ\ast})$.  These are used as an intermediate step between $\alpha_i^\ast$ and $\hat{\alpha}_i$ or between $\bm{\theta}_i^\ast$ and $\hat{\bm{\theta}}_i$ (or $\bm{\theta}^\ast$ and $\hat{\bm{\theta}}$) in the lemmas that follow.

\begin{lem} \label{lem:4K}
Under the assumptions in Theorem~\ref{thm:boot},
\begin{multline*}
    \frac{1}{N} \sum_{i=1}^N \left( \mathbb{K}_{i}^{\theta*}(\bm{\theta}_{i}^*) - \mathbb{K}_{i}^{\theta*}(\hat{\bm{\theta}}_i) - \exs{\mathbb{K}_{i}^{\theta*}(\bm{\theta}^*_i) + \mathbb{K}_{i}^{\theta*}(\hat{\bm{\theta}}_i)} \right) \\
    = \Ops{ T^{-1/2} \log T \|\bm{\beta}^\ast - \hat{\bm{\beta}}\|^{1/2} + T^{-1} (\log T)^2 + N^{-1/2} T^{-2/3}}.
\end{multline*}
\end{lem}

\begin{proof}[Proof of Lemma~\ref{lem:4K}]

and recall the definition of $\alpha_i^{\circ\ast}$ from Lemma~\ref{lem:alphas}.  Note that
  \begin{equation} \label{4Ksum}
    \frac{1}{N} \sum_{i=1}^N \mathcal{K}_i(\bm{\theta}_i^\ast, \hat{\bm{\theta}}_i) = \frac{1}{N} \sum_{i=1}^N \mathcal{K}_i(\bm{\theta}_i^\ast, \bm{\theta}_i^{\circ\ast}) + \frac{1}{N} \sum_{i=1}^N \mathcal{K}_i(\bm{\theta}_i^{\circ\ast}, \hat{\bm{\theta}}_i).
  \end{equation}

  Recalling definitions~\eqref{def:Gdelta} and~\eqref{def:empproc} and defining $\check{M} = M(1 + \overline{f} / \underline{f})$, which bounds $\check{\bm{x}}_{it}$, the first term on the right-hand side of~\eqref{4Ksum} can be bounded by noting that 
  \begin{align}
    \sup_i \left| \mathcal{K}_i(\bm{\theta}_i^\ast, \bm{\theta}_i^{\circ\ast}) \right| &\leq \sup_i 2\check{M} \|\mathbb{P}_{iT}^\ast - \mathbb{P}_{iT}\|_{\mathcal{G}(\|\bm{\beta}^\ast - \hat{\bm{\beta}}\| + |\alpha^\ast - \alpha_i^{\circ\ast}|)} \notag \\
    {} &= \Ops{ \chi_T(\|\bm{\beta}^\ast - \hat{\bm{\beta}}\| + |\alpha^\ast - \alpha_i^{\circ\ast}|) } \notag \\
    {} &= \Ops{ \left( \|\bm{\beta}^\ast - \hat{\bm{\beta}}\| + |\alpha^\ast - \alpha_i^{\circ\ast}| \right)^{1/2} T^{-1/2} \log T + T^{-1} (\log T)^2 } \notag \\
    \intertext{and Lemma~\ref{lem:alphas} implies}
    {} &= \Ops{ \|\bm{\beta}^\ast - \hat{\bm{\beta}}\|^{1/2} T^{-1/2} \log T + T^{-1} (\log T)^2 }. \label{4KorderA}
  \end{align}

  The second term on the right-hand side of~\eqref{4Ksum} is more involved.  Lemma~\ref{lem:4Kexp} shows that
  \begin{equation*}
    \sup_i \left\lVert \exs{ \mathcal{K}_i(\bm{\theta}_i^{\circ\ast}, \hat{\bm{\theta}}_i) } \right\rVert = \Op{T^{-1} (\log T)^2}.
  \end{equation*}
  Furthermore, we know that
  \begin{align*}
    \left\lVert \mathcal{K}_i(\bm{\theta}_i^{\circ\ast}, \hat{\bm{\theta}}_i) \right\rVert &\leq 2\check{M} \|\mathbb{P}_{iT}^\ast - \mathbb{P}_{iT}\|_{\mathcal{G}(|\alpha_i^{\circ\ast} - \hat{\alpha}_i|)}, \\
    \left\lVert \mathcal{K}_i(\bm{\theta}_i^{\circ\ast}, \hat{\bm{\theta}}_i) \right\rVert &\leq 4\check{M} \; a.s.,
  \end{align*}
  and these terms are independent over $i$ (recall that $\hat{\bm{\theta}}$ is fixed under the bootstrap measure).  An argument identical to that made for Lemma~\ref{lem:alphasmall} but conditional on the data and $\hat{\bm{\beta}}$ (instead of estimating $\bm{\beta}^\ast$ along with the $\alpha_i^\ast$ terms) implies that
  \begin{equation} \label{acirc_small}
    \sup_i |\alpha_i^{\circ\ast} - \hat{\alpha}_i| = \Ops{T^{-1/2} (\log T)^{1/2}}.
  \end{equation}
  Then
  \begin{align*}
    {} &\phantom{=} \probs{\sup_i \|\mathcal{K}_i(\bm{\theta}_i^{\circ\ast}, \hat{\bm{\theta}}_i)\| > T^{-2/3}} \\
    {} &\leq \probs{\sup_i |\alpha_i^{\circ\ast} - \hat{\alpha}_i| > C T^{-1/2} (\log T)^{1/2}} \\
    {} &\phantom{=} \qquad \qquad \qquad + \probs{2\check{M} \|\mathbb{P}_{iT}^\ast - \mathbb{P}_{iT}\|_{\mathcal{G}(cT^{-1/2}(\log T)^{1/2})} > T^{-2/3}} \\
    {} &= \Op{T^{-2}},
  \end{align*}
  where we use~\eqref{acirc_small} and $\chi_T(CT^{-1/2} (\log T)^{1/2}) = C^{1/2} T^{-3/4} (\log T)^{5/4} + T^{-1} (\log T)^2 = o(T^{-2/3})$.

  Then letting $\mathcal{K}_{ij}$ denote the $j$-th coordinate of $\mathcal{K}_i$,
  \begin{align*}
    {} &\phantom{=} \sup_{i,j} \text{Var}^\ast \left( \mathcal{K}_{ij}(\bm{\theta}_i^{\circ\ast}, \hat{\bm{\theta}}_i) \right) \\
    {} &\leq \sup_{i,j} \exs{\mathcal{K}_{ij}^2(\bm{\theta}_i^{\circ\ast}, \hat{\bm{\theta}}_i)} \\
    {} &= \sup_{i,j} \exs{ \mathcal{K}_{ij}^2(\bm{\theta}_i^{\circ\ast}, \hat{\bm{\theta}}_i) \left( I(|\mathcal{K}_{ij}(\bm{\theta}_i^{\circ\ast}, \hat{\bm{\theta}}_i)| > T^{-2/3}) + I(|\mathcal{K}_{ij}(\bm{\theta}_i^{\circ\ast}, \hat{\bm{\theta}}_i)| \leq T^{-2/3}) \right)} \\
    {} &\leq 16\check{M}^2 \sup_{i,j} \probs{ |\mathcal{K}_{ij}(\bm{\theta}_i^{\circ\ast}, \hat{\bm{\theta}}_i)| > T^{-2/3} } + T^{-4/3} \\
    {} &= \Op{T^{-4/3}}.
  \end{align*}

  Therefore
  \begin{align}
    \frac{1}{N} \sum_{i=1}^N \mathcal{K}_i(\bm{\theta}_i^{\circ\ast}, \hat{\bm{\theta}}_i) &= \Ops{ \sup_{i,j} \left| \exs{\mathcal{K}_{ij}(\bm{\theta}_i^{\circ\ast}, \hat{\bm{\theta}}_i)} + \sqrt{\text{Var}^\ast(\mathcal{K}_{ij}(\bm{\theta}_i^{\circ\ast}, \hat{\bm{\theta}}_i)) / N} \right| } \notag \\
    {} &= \Ops{T^{-1} (\log T)^2 + N^{-1/2} T^{-2/3}}. \label{4KorderB}
  \end{align}

  Then put~\eqref{4KorderA} and~\eqref{4KorderB} in~\eqref{4Ksum} to find the result.

\end{proof}

\begin{lem} \label{lem:alphas}
  Under the assumptions of Theorem~\ref{thm:boot}, the estimators $\alpha_i^{\circ\ast}$, $i = 1, \ldots, N$ defined in~\eqref{def:alphas} satisfy
  \begin{equation*}
      \sup_i |\alpha_i^\ast - \alpha_i^{\circ\ast}| = \Ops{\|\bm{\beta}^\ast - \hat{\bm{\beta}}\| + T^{-1} (\log T)^2}
  \end{equation*}
  as $N, T \rightarrow \infty$.
\end{lem}

\begin{proof}[Proof of Lemma~\ref{lem:alphas}]
This follows the argument for GGV Lemma 7 with alterations made for the bootstrap distribution.  For this proof, define the functions
\begin{equation*}
    \hat{\mathbb{F}}_{iT}^\ast(y, \bm{\beta}) = \frac{1}{T} \sum_{t=1}^T I(y_{it}^\ast - \bm{x}_{it}\pr \bm{\beta} \leq y), \qquad \mathbb{F}_{iT}^\ast(y, \bm{\beta}) = \exs{\hat{\mathbb{F}}_{iT}^\ast(y, \bm{\beta})},
\end{equation*}
  and recall $\chi_T(\delta)$ defined in~\eqref{def:chi}.

  Under Assumptions~\ref{b:Gsupp}, \ref{a:minf} and \ref{a:boundedX}, for $\bm{\beta}$ close to $\hat{\bm{\beta}}$, using the function class $\mathcal{G}(\delta)$ defined in~\eqref{def:Gdelta} and result~\eqref{GGVL5},
\begin{align*}
  {} &\phantom{=} \sup_i \left\lvert \hat{\mathbb{F}}_{iT}^\ast(y, \bm{\beta}) - \hat{\mathbb{F}}_{iT}^\ast(y, \hat{\bm{\beta}}) \right\rvert \\
  {} &= \sup_i \ex{\exs{\left\lvert \hat{\mathbb{F}}_{iT}^\ast(y, \bm{\beta}) - \hat{\mathbb{F}}_{iT}^\ast(y, \hat{\bm{\beta}}) \right\rvert } \rvert \bm{X}_i } + \Op{\chi_T(\|\bm{\beta} - \hat{\bm{\beta}}\|)} \\
    {} &= \sup_i \exs{\ex{\left\lvert \hat{\mathbb{F}}_{iT}^\ast(y, \bm{\beta}) - \hat{\mathbb{F}}_{iT}^\ast(y, \hat{\bm{\beta}}) \right\rvert } \rvert \bm{X}_i} + \Op{ \chi_T(\|\bm{\beta} - \hat{\bm{\beta}}\|) } \\
    {} &= \exs{ F_i((y + \bm{x}_{it}\pr (\bm{\beta} - \hat{\bm{\beta}}) - \hat{\alpha}_i) / w_{it}) - F_i((y - \hat{\alpha}_i) / w_{it}) } + \Op{ \chi_T(\|\bm{\beta} - \hat{\bm{\beta}}\|) } \\
    {} &= \exs{ f_i((y - \hat{\alpha}_i) / w_{it}) \bm{x}_{it}\pr (\bm{\beta} - \hat{\bm{\beta}}) / w_{it} } + \Op{\|\bm{\beta} - \hat{\bm{\beta}}\|^2} + \Op{ \chi_T(\|\bm{\beta} - \hat{\bm{\beta}}\|) } \\
    {} &\lesssim \frac{\overline{f} M}{c_1 \wedge c_2} \|\bm{\beta} - \hat{\bm{\beta}}\| + \chi_T(\|\bm{\beta} - \hat{\bm{\beta}}\|).
\end{align*}

Similarly, for $y, y\pr$ we know that
\begin{equation*}
\sup_i \left\lvert \hat{\mathbb{F}}_{iT}^\ast(y, \hat{\bm{\beta}}) - \hat{\mathbb{F}}_{iT}^\ast(y\pr, \hat{\bm{\beta}}) - \ex{\exs{\hat{\mathbb{F}}_{iT}^\ast(y, \hat{\bm{\beta}}) - \hat{\mathbb{F}}_{iT}^\ast(y\pr, \hat{\bm{\beta}})}} \right\rvert \lesssim \chi_T(|y - y\pr|).
\end{equation*}
Because the estimates $\alpha_i^\ast$ and $\alpha_i^{\circ\ast}$ are empirical quantiles from the within-$i$ residuals $y_{it}^\ast - \bm{x}_{it}\tr \bm{\beta}$ (for different $\bm{\beta}$) we may write
\begin{align*}
    \frac{2}{T} &\geq \left\lvert \hat{\mathbb{F}}_{iT}^\ast(\alpha_i^\ast, \bm{\beta}^\ast) - \hat{\mathbb{F}}_{iT}^\ast(\alpha_i^{\circ\ast}, \hat{\bm{\beta}}) \right\rvert \\
    {} &\gtrsim \left\lvert \hat{\mathbb{F}}_{iT}^\ast(\alpha_i^\ast, \hat{\bm{\beta}}) - \hat{\mathbb{F}}_{iT}^\ast(\alpha_i^{\circ\ast}, \hat{\bm{\beta}}) \right\rvert - \chi_T(\|\bm{\beta}^\ast - \hat{\bm{\beta}}\|) - \frac{\overline{f} M}{c_1 \wedge c_2} \|\bm{\beta}^\ast - \hat{\bm{\beta}}\| \\
    {} &\gtrsim \left\lvert \ex{\exs{\hat{\mathbb{F}}_{iT}^\ast(\alpha_i^\ast, \hat{\bm{\beta}}) - \hat{\mathbb{F}}_{iT}^\ast(\alpha_i^{\circ\ast}, \hat{\bm{\beta}}) }} \right\rvert - \chi_T(\|\bm{\beta}^\ast - \hat{\bm{\beta}}\|) - \frac{\overline{f} M}{c_1 \wedge c_2} \|\bm{\beta}^\ast - \hat{\bm{\beta}}\| - \chi_T(|\alpha_i^\ast - \alpha_i^{\circ\ast}|).
\end{align*}
The last steps are the same as in GGV Lemma 7 and lead to the result.
\end{proof}

In the proof of Lemma~\ref{lem:4Kexp} and the supporting Lemma~\ref{lem:LOO}, we use leave-one-cell quantities.  We will use
\begin{equation} \label{al_LOO}
  \alpha_i^{(-j)} = \argmin_a \sum_{j' \neq j} \sum_{s=1}^l \rho_\tau(y_{ij's}^\ast - \bm{x}_{ij's}'\hat{\bm{\beta}} - a),
\end{equation}
and let $\bm{\theta}_i^{(-j)} = (\hat{\bm{\beta}}', \alpha_i^{(-j)})'$.

\begin{lem}\label{lem:4Kexp}
  Under the assumptions in Theorem~\ref{thm:boot},
  \begin{equation*}
    \sup_i \left\lVert \exs{ \mathcal{K}_i(\bm{\theta}_i^{\circ\ast}, \hat{\bm{\theta}}_i) } \right\rVert = \Op{T^{-1} (\log T)^2}.
  \end{equation*}
  where $\mathcal{K}_i$ is defined in~\eqref{kcal}.
\end{lem}

\begin{proof}[Proof of Lemma~\ref{lem:4Kexp}]
  This proof follows the proof of equation~(29) in GGV Lemma 6. In this proof we need to give a name to the following expected value function.  For a fixed $\bm{\theta}_i \in \R^{p+1}$, let
  \begin{equation*}
    K_i^{\theta}(\bm{\theta}_i) = \exs{\mathbb{K}_i^{\theta\ast}(\bm{\theta}_i)} = \frac{1}{T} \sum_{t=1}^T \check{\bm{x}}_{it} \left( \tau - \probs{ w_{it} |\hat{u}_{it}| \leq \bm{z}_{it}'(\bm{\theta} - \hat{\bm{\theta}}_i) } \right).
  \end{equation*}
  The reason is that in the expressions below, some expected values are computed and then evaluated at a parameter estimate, and some expected values (which may be different) involve the evaluation of the function and the estimate simultaneously under the bootstrap measure.

  For any $i$ and any $j = 1, \ldots, b$, recalling definition~\eqref{kcal}, write
\begin{equation*}
  \mathcal{K}_i(\bm{\theta}_i^{\circ\ast}, \hat{\bm{\theta}}_i) = \mathcal{K}_i(\bm{\theta}_i^{(-j)}, \hat{\bm{\theta}}_i) - \mathcal{K}_i(\bm{\theta}_i^{(-j)}, \bm{\theta}_i^{\circ\ast}).
\end{equation*}

  Neither $\hat{\bm{\theta}}_i$ nor $\bm{\theta}_i^{(-j)}$ depend on any cell-$j$ weights in the first term of the right-hand side and therefore the first term does not vary under the boostrap measure and has mean zero for all $j$, implying
  \begin{equation*}
    \sup_{i,j} \left\| \exs{\mathcal{K}_i(\bm{\theta}_i^{(-j)}, \hat{\bm{\theta}}_i)} \right\| = 0.
  \end{equation*}
  Therefore we may write, for any $j$,
  \begin{align}
    \exs{ \mathcal{K}_i(\bm{\theta}_i^{\circ\ast}, \hat{\bm{\theta}}_i) } &= -\exs{ \mathcal{K}_i(\bm{\theta}_i^{(-j)}, \bm{\theta}_i^{\circ\ast}) } \notag \\
    {} &= -\exs{ \mathbb{K}_i^{\theta\ast}(\bm{\theta}_i^{(-j)}) - \mathbb{K}_i^{\theta\ast}(\bm{\theta}_i^{\circ\ast})} + \exs{ K_i^{\theta}(\bm{\theta}_i^{(-j)}) - K_i^{\theta}(\bm{\theta}_i^{\circ\ast}) }. \label{E4K_LOO}
\end{align}
  Next we bound two pieces on the right-hand side of~\eqref{E4K_LOO}.

  First, for a sequence $\delta_T \searrow 0$, recall $\check{\bm{x}}_{it} = \bm{x}_{it} - \bar{\varphi}_i^{-1} \bar{\bm{g}}_i$ and let $\check{M} = M(1 + \overline{f} / \underline{f})$ (noting that this bounds $\|\check{\bm{x}}_{it}\|$) and write
  \begin{multline} \label{partway}
    \left\lVert \exs{ \mathbb{K}_i^{\theta\ast}(\bm{\theta}_i^{(-j)}) - \mathbb{K}_i^{\theta\ast}(\bm{\theta}_i^{\circ\ast}) } \right\rVert \\
    {} \leq \check{M} \left| \exs{ \frac{1}{T} \sum_{j=1}^b \sum_{s=1}^l \left( I(y_{ijs}^\ast \leq \bm{z}_{ijs}' \bm{\theta}_i^{(-j)}) - I(y_{ijs}^\ast \leq \bm{z}_{ijs}'\bm{\theta}_i^{\circ\ast}) \right) } \right| \\
    {} \leq \check{M} \frac{1}{T} \sum_{j=1}^b \sum_{s=1}^l \text{E}^\ast \Bigg[ I\left( |y_{ijs}^\ast - \bm{z}_{ijs}' \bm{\theta}_i^{(-j)}| \leq (M+1) \delta_T \right) \\
    {} \times I\left( \sup_{j} |\bm{z}_{ijs}' (\bm{\theta}_i^{(-j)} - \bm{\theta}_i^{\circ\ast})| \leq (M+1) \delta_T \right) \Bigg] \\
    {} + \check{M} \probs{ \sup_{j} |\bm{z}_{ijs}' (\bm{\theta}_i^{(-j)} - \bm{\theta}_i^{\circ\ast})| > (M+1) \delta_T }.
  \end{multline}

  Tedious calculations that are similar to the steps in Lemma~\ref{lem:boot_cons2} lead to
  \begin{multline}
    \frac{1}{T} \sum_{j=1}^b \sum_{s=1}^l \text{E}^\ast \Bigg[ I\left( |y_{ijs}^\ast - \bm{z}_{ijs}' \bm{\theta}_i^{(-j)}| \leq (M+1) \delta_T \right) \times I\left( \sup_{j} |\bm{z}_{ijs}' (\bm{\theta}_i^{(-j)} - \bm{\theta}_i^{\circ\ast})| \leq (M+1) \delta_T \right) \Bigg] \\
    = 2(M+1) \delta_T \frac{1}{T} \sum_{t=1}^T f_i(0 | \bm{x}_{it}) + \Op{ \sup_i \|\hat{\bm{\theta}}_i - \bm{\theta}_{i0}\|^2 } \\
    + \Ops{ \delta_T^2 } + \Ops{\chi_T(\delta_T)} \\
    = 2(M+1) \delta_T \frac{1}{T} \sum_{t=1}^T f_i(0 | \bm{x}_{it}) + \Op{ \|\hat{\bm{\theta}}_i - \bm{\theta}_{i0}\|^2 } + \Ops{\chi_T(\delta_T)} \label{hard_expansion}
  \end{multline}
  Now put~\eqref{hard_expansion} in~\eqref{partway} to write
  \begin{multline}
    \left\lVert \exs{ \mathbb{K}_i^{\theta\ast}(\bm{\theta}_i^{(-j)}) - \mathbb{K}_i^{\theta\ast}(\bm{\theta}_i^{\circ\ast}) } \right\rVert \\
    = \Ops{\delta_T + \chi_T(\delta_T)} + \Op{ \| \hat{\bm{\theta}}_i - \bm{\theta}_{i0} \|^2 } \\
    + \check{M} \probs{ \sup_j |\bm{z}_{ijs}' (\bm{\theta}_i^{(-j)} - \bm{\theta}_i^{\circ\ast})| > (M+1) \delta_T }. \label{lem_pt1}
  \end{multline}

  On the other hand, for any $\bm{\theta}_1$ and $\bm{\theta}_2$ in a $\delta_T$-neighborhood of $\hat{\bm{\theta}}_i$ still more tedious calculations similar to Lemma~\ref{lem:boot_cons2} result in
\begin{align*}
  \left\lVert \exs{K_i^{\theta}(\bm{\theta}_1) - K_i^{\theta}(\bm{\theta}_2)} \right\rVert &= \left\lVert \check{\bm{x}}_{it} \exs{ \psi_\tau(y_{it}^\ast - \bm{z}_{it}\tr \bm{\theta}_1) - \psi_\tau(y_{it}^\ast - \bm{z}_{it}\tr \bm{\theta}_2) } \right\rVert \\
  {} &\leq \check{M} \overline{f} |\bm{z}_{it}\tr (\bm{\theta}_1 - \bm{\theta}_2)| \\
  {} &\phantom{=} \qquad \qquad + \Ops{\|\bm{\theta}_1 - \bm{\theta}_2\|} + \Op{ \|\hat{\bm{\theta}}_i - \bm{\theta}_{i0}\|^2 } \\
  {} &\leq \check{M} \overline{f} (M+1) \|\bm{\theta}_1 - \bm{\theta}_2\| \\
  {} &\phantom{=} \qquad \qquad + \Ops{\|\bm{\theta}_1 - \bm{\theta}_2\|} + \Op{ \|\hat{\bm{\theta}}_i - \bm{\theta}_{i0}\|^2 }.
\end{align*}
  Also we know that $\|\mathbb{K}_i^{\theta\ast}(\bm{\theta}_1) - \mathbb{K}_i^{\theta\ast}(\bm{\theta}_2)\| \leq 2\check{M}$.  This implies
\begin{multline}
  \left\lVert \exs{ K_i^{\theta}(\bm{\theta}_i^{(-j)}) - K_i^{\theta}(\bm{\theta}_i^{\circ\ast}) } \right\rVert \\
  {} \leq \check{M}\overline{f}(M+1) \delta_T + 2\check{M} \probs{\sup_j \| \bm{\theta}_i^{(-j)} - \bm{\theta}_i^{\circ\ast} \| > \delta_T }. \label{lem_pt2}
\end{multline}

Putting bounds~\eqref{lem_pt1} and~\eqref{lem_pt2} into~\eqref{E4K_LOO} we find that with high probability,
\begin{multline}
    \left\| \exs{ \mathbb{K}_i^{\theta\ast}(\bm{\theta}_i^{(-j)}) - \mathbb{K}_i^{\theta\ast}(\bm{\theta}_i^{\circ\ast})} + \exs{ K_i^{\theta}(\bm{\theta}_i^{(-j)}) - K_i^{\theta}(\bm{\theta}_i^{\circ\ast}) } \right\| \\
  = \Ops{ \delta_T + \chi_T(\delta_T) } + \Op{ \|\hat{\bm{\theta}}_i - \bm{\theta}_{i0}\|^2 } + \Op{ \probs{ \sup_j \| \bm{\theta}_i^{(-j)} - \bm{\theta}_i^{\circ\ast} \| > \delta_T } } .
\end{multline}
Finally, choose $\delta_T = CT^{-1} (\log T)^2$.  GGV Theorems 1 and 3 show that $\|\hat{\bm{\theta}}_i - \bm{\theta}_{i0}\|^2 = \Op{ (NT)^{-1/2} + T^{-1} (\log T)^2} = \Op{T^{-1}(\log T)^2}$ under the restriction on $N$ and $T$ assumed at the end of the statement of Lemma~\ref{lem:feqr}.  Then with the chosen $\delta_T$, Lemma~\ref{lem:LOO} implies the last term is $\Op{NT^{-\kappa}} = \op{NT^{-1}(\log T)^2}$ by choosing $\kappa$ large enough and for large $T$ under the restriction on $N$ and $T$ assumed at the end of Lemma~\ref{lem:feqr}.
\end{proof}

\begin{lem} \label{lem:LOO}
  Under the assumptions of Theorem~\ref{thm:boot}, for constants $\kappa$ and $C$, when $T$ is large,
\end{lem}
\begin{equation*}
  \probs{ \sup_{i, j} \| \bm{\theta}_i^{(-j)} - \bm{\theta}_i^{\circ\ast} \| \geq C \kappa^4 T^{-1} (\log T)^2 } \leq 2 N T^{-\kappa}.
\end{equation*}

\begin{proof}[Proof of Lemma~\ref{lem:LOO}]
  Using $\alpha_i^{(-j)}$ as defined in~\eqref{al_LOO} and the corresponding $\bm{\theta}_i^{(-j)}$, define leave-one-cell-out score functions (recall cells are of size $l$)
  \begin{equation*}
    \mathbb{K}_i^{(-j)}(\bm{\theta}) = \frac{1}{T-l} \sum_{j' \neq j} \sum_{s=1}^l \check{\bm{x}}_{ij's} \psi_\tau(y_{ij's}^\ast - \bm{z}_{ij's}\tr \bm{\theta})
  \end{equation*}
  and
  \begin{equation*}
    \mathcal{K}_i^{(-j)}(\bm{\theta}_1, \bm{\theta}_2) = \mathbb{K}_{i}^{(-j)}(\bm{\theta}_1) - \mathbb{K}_{i}^{(-j)}(\bm{\theta}_2) - \exs{\mathbb{K}_i^{(-j)}(\bm{\theta}_1)} + \exs{\mathbb{K}_i^{(-j)}(\bm{\theta}_2)}.
  \end{equation*}

  Use these to derive expansions for each $j = 1, \ldots, b$:
  \begin{multline}
    \bm{\theta}_i^{(-j)} - \hat{\bm{\theta}}_i = -\bar{\bm{J}}_i^{-1} \mathbb{K}_{i}^{(-j)}(\hat{\bm{\theta}}_i) + \bar{\bm{J}}_i^{-1} \mathbb{K}_{i}^{(-j)}(\bm{\theta}_i^{(-j)}) - \bar{\bm{J}}_i^{-1} \mathcal{K}_{i}^{(-j)}(\bm{\theta}_i^{(-j)}, \hat{\bm{\theta}}_i) \\
    - \bar{\bm{J}}_i^{-1} \left( \exs{\mathbb{K}_i^{(-j)}(\bm{\theta}_i^{(-j)})} - \exs{\mathbb{K}_i^{(-j)}(\hat{\bm{\theta}}_i)} - \bar{\bm{J}}_i \left( \bm{\theta}_i^{(-j)} - \hat{\bm{\theta}}_i \right) \right).
  \end{multline}

  The right-hand side is a sum of four terms: call each one $r_{ik}^{(-j)}$ for $k = 1, 2, 3, 4$.  There are analogous whole-sample terms (i.e., that do not leave cell $j$ out) that are composed of $\mathbb{K}_i^{\theta\ast}$ functions.  Denote these as $r_{ik}$.  Then the maximal difference between any one leave-one-cell out estimator and $\bm{\theta}_i^{\circ\ast}$ is
  \begin{align*}
    \sup_{i\,j} \| \bm{\theta}_i^{(-j)} - \bm{\theta}_i^{\circ\ast} \| &= \sup_{i,j} \| (\bm{\theta}_i^{(-j)} - \hat{\bm{\theta}}_i) - (\bm{\theta}_i^{\circ\ast} - \hat{\bm{\theta}}_i) \| \\
    {} &\leq \sup_{i,j} \lambda_{\text{min}}(\bar{\bm{J}}_i) \sum_{k=1}^4 \left\| r_{ik}^{(-j)} - r_{ik} \right\|.
  \end{align*}
    $\lambda_{\text{min}}(\bar{\bm{J}}_i)$ is bounded by Assumption~\ref{assume:Avar_dependent} and we ignore it below.  Next, we deal with the terms in the sum in order.

  Since the summands of $r_{i1}^{(-j)}$ and $r_{i1}$ are averages of terms that are bounded and have mean zero under the bootstrap measure (recall there are $l$ terms in $b$ cells and $bl = T$), for each $j$ we have
  \begin{equation*}
    \left\| r_{i1}^{(-j)} - r_{i1} \right\| = \left( \frac{1}{T-l} - \frac{1}{T} \right) \sum_{j' \neq j} \sum_{s=1}^l \check{\bm{x}}_{ij's} \psi_\tau(y_{ij's} - \bm{z}_{ij's}\tr \hat{\bm{\theta}}) + \frac{1}{T-l} \sum_{s=1}^l \check{\bm{x}}_{ijs} \psi_\tau(y_{ijs} - \bm{z}_{ijs}\tr \hat{\bm{\theta}}).
  \end{equation*}
  However, using Hoeffding's inequality we can see that the first part satisfies
  \begin{align*}
    &\phantom{=} \probs{ \sup_j \frac{l}{T} \frac{1}{T-l} \sum_{j' \neq j} \sum_{s=1}^l \check{\bm{x}}_{ij's} \psi_\tau(y_{ij's} - \bm{z}_{ij's}\tr \hat{\bm{\theta}}) > C (T-l)^{-3/2} (\log (T-l))^{1/2} } \\
    {} &\leq \sum_{j=1}^b \probs{ \frac{l}{T} \frac{1}{T-l} \sum_{j' \neq j} \sum_{s=1}^l \check{\bm{x}}_{ij's} \psi_\tau(y_{ij's} - \bm{z}_{ij's}\tr \hat{\bm{\theta}}) > C (T-l)^{-3/2} (\log (T-l))^{1/2} } \\
    {} &\leq 2 b (T-l)^{-C^2/2\check{M}^2 l^2} \sim \frac{2}{l} T^{1 - C^2/2\check{M}^2l^2}.
  \end{align*}
  That is,
  \begin{equation} \label{diff_r1}
    \sup_j \left\| r_{i1}^{(-j)} - r_{i1} \right\| = \Ops{T^{-3/2} (\log T)^{1/2}} + \Ops{T^{-1}} = \Ops{T^{-1}}.
  \end{equation}
  The second terms are small because they both optimize the quantile regression objective function, that is,
  \begin{equation} \label{diff_r2}
    \sup_{i,j} \left\| r_{i2}^{(-j)} - r_{i2} \right\| = \sup_{i,j} \left\| \mathbb{K}_{i}^{(-j)}(\bm{\theta}_i^{(-j)}) - \mathbb{K}_{i}^{\theta\ast}(\bm{\theta}_i^{\circ\ast}) \right\| \leq C \left( \frac{1}{T-l} + \frac{1}{T} \right) = O(T^{-1}) \; a.s.
  \end{equation}
  To bound the third difference, consider the event
  \begin{equation*}
    \mathcal{A}_i(\kappa) = \left\{ \sup_{\delta \in [0, 1]} \sup_{\|\bm{\theta}_1 - \bm{\theta}_2\| \leq \delta} \sup_j \frac{\| \mathcal{K}_{i}^{(-j)}(\bm{\theta}_1, \bm{\theta}_2) \|}{\chi_T(\delta)} \leq 2 C \kappa^2 \right\}.
  \end{equation*}
  Because only one cell of $l$ terms is missing from $\mathcal{K}_i^{(-j)}$,
  \begin{equation*}
    \left\| \sup_{\|\bm{\theta}_1 - \bm{\theta}_2\| \leq \delta} \sup_j \| \mathcal{K}_{i}^{(-j)}(\bm{\theta}_1, \bm{\theta}_2) \| - \sup_{\|\bm{\theta}_1 - \bm{\theta}_2\| \leq \delta} \| \mathcal{K}_{i}(\bm{\theta}_1, \bm{\theta}_2) \| \right\| = O(T^{-1}) \; a.s.,
  \end{equation*}
  and~\eqref{GGVL5} implies that
  \begin{equation*}
    \inf_i \probs{ \mathcal{A}_i(\kappa) } \geq \inf_i \probs{ \sup_{\delta \in [0, 1]} \sup_{\|\bm{\theta}_1 - \bm{\theta}_2\| \leq \delta} \frac{\| \mathcal{K}_i(\bm{\theta}_1, \bm{\theta}_2) \|}{\chi_T(\delta)} \leq 2 C \kappa^2 } \geq 1 - T^{-\kappa},
  \end{equation*}
  so that assuming that $\sup_i \|\bm{\theta}_i^{(-j)} - \bm{\theta}_i^{\circ\ast}\| = \ops{1}$ (which is verified below),
  \begin{equation} \label{diff_r3}
    \sup_{i,j} \left\| r_{i3}^{(-j)} - r_{i3} \right\| = \Ops{ \chi_T( \sup_{i,j} \|\bm{\theta}_i^{(-j)} - \bm{\theta}_i^{\circ\ast}\|) } + O(T^{-1}).
  \end{equation}
  Finally, to bound the fourth, calculations similar to the calculations in Lemma~\ref{lem:boot_cons2} imply that
  \begin{align*}
    \exs{\mathbb{K}_i^{(-j)}(\bm{\theta}_1)} - \exs{\mathbb{K}_i^{(-j)}(\bm{\theta}_2)} &= \frac{1}{T-l} \sum_{j' \neq j} \sum_{s=1}^l \check{x}_{ij's} \exs{ I(y_{ij's}^\ast \leq z_{ij's}'\bm{\theta}_2) - I(y_{ij's}^\ast \leq z_{ij's}'\bm{\theta}_1) } \\
    {} &= \frac{1}{T-l} \sum_{j',s} f(0 | \bm{z}_{ij's}) \check{\bm{x}}_{ij's} \check{\bm{x}}_{ij's}' (\bm{\theta}_1 - \bm{\theta}_2) \\
    {} &\phantom{=} \qquad \qquad + \Ops{\| \bm{\theta}_1 - \bm{\theta}_2 \|^2} + \Op{ \sup_i \| \hat{\bm{\theta}}_i - \bm{\theta}_{i0} \|^2 }.
  \end{align*}
  As $T \rightarrow \infty$, the term premultiplying $\bm{\theta}_1 - \bm{\theta}_2$ converges in probability conditionally to $\bar{\bm{J}}_i$.  Then assuming that $\|\bm{\theta}_1 - \bm{\theta}_2\| \leq \delta$ for $\delta$ small enough,
  \begin{multline} \label{diff_taylor}
    \sup_i \left\| \exs{\mathbb{K}_i^{(-j)}(\bm{\theta}_1)} - \exs{\mathbb{K}_i^{(-j)}(\bm{\theta}_2)} - \bar{\bm{J}}_i \left( \bm{\theta}_1 - \bm{\theta}_2 \right) \right\| \\
    = \Ops{\delta^2} + \Op{ \sup_i \| \hat{\bm{\theta}}_i - \bm{\theta}_{i0} \|^2 } + \Op{T^{-1}},
  \end{multline}
  where the last term accounts for the difference between the first matrix and $\bar{\bm{J}}_i$, which differ only on $l$ terms out of $T$.  We have an analogous expression using $\mathbb{K}_i^{\circ\ast}$ functions instead of $\mathbb{K}_i^{(-j)}$ functions.
  Lemma S.1.2 of \citet{ChaoVolgushevCheng17} implies that because $\bm{\theta}_i^{(-j)}$ is the minimizer of a quantile regression objective, the size of its difference from $\hat{\bm{\theta}}_i$ is related to the size of the score function: for each $i,j$,
  \begin{equation*}
    \left\{ \|\mathbb{K}_i^{(-j)}(\hat{\bm{\theta}}_i)\| \leq C_1 \right\} \subseteq \left\{ \| \bm{\theta}_i^{(-j)} - \hat{\bm{\theta}}_i \| \leq C_2 \|\mathbb{K}_i^{(-j)}(\hat{\bm{\theta}}_i)\| \right\}.
  \end{equation*}
  Also note that~\eqref{diff_r1} implies that the scores evaluated at $\hat{\bm{\theta}}_i$ are within $\Ops{T^{-1}}$ of each other.  Defining the events
  \begin{equation*}
    \mathcal{B}_i(\kappa) = \left\{ \left\| \mathbb{K}_i^{\theta\ast}(\hat{\bm{\theta}}_i) \right\| \leq 2 C \kappa^{1/2} T^{-1/2} (\log T)^{1/2} \right\},
  \end{equation*}
  GGV Lemma 5 implies that $\sup_i \prob{\mathcal{B}_i(\kappa)} \geq 1 - T^{-\kappa}$ and it implies analogous events for $\|\mathbb{K}_i^{(-j)}(\hat{\bm{\theta}}_i)\|$.  Therefore on $\mathcal{B}_i(\kappa)$, the expansion~\eqref{diff_taylor} implies
  \begin{equation} \label{diff_r4}
    \sup_{i,j} \left\| r_{i4}^{(-j)} - r_{i4} \right\| \leq C\kappa T^{-1} \log T.
  \end{equation}
  This used the fact that $\|\hat{\bm{\theta}}_i - \bm{\theta}_{i0}\| = \Op{(NT)^{-1/2} + T{^-1}(\log T)^2}$, as shown in the proof of GGV Theorems 1 and 3.

  Next, add~\eqref{diff_r1}, \eqref{diff_r2}, \eqref{diff_r3} and~\eqref{diff_r4} to find that on $\mathcal{A}_i(\kappa) \cap \mathcal{B}_i(\kappa)$, that is, with probability at least $1 - T^{-\kappa}$, letting $\delta_{iT} = \sup_j \| \bm{\theta}_i^{(-j)} - \bm{\theta}_i^{\circ\ast} \|$,
  \begin{align*}
    \delta_{iT} &\leq C \kappa^2 \left( T^{-1} + T^{-1/2} \log T \delta_T^{1/2} + T^{-1} (\log T)^2 \right) \\
    \intertext{and using the implication $\delta \leq a + b\delta^{1/2} \Rightarrow \delta \leq 4\max\{a, b^2\}$,}
    \delta_{iT} &\leq 4C \kappa^4 T^{-1} (\log T)^2.
  \end{align*}
  Then the union bound implies the result.
\end{proof}

\section{Other bootstrap methods for dependent data}\label{sec:ETBB}

This section briefly describes the other bootstrap approaches for inference for conditional quantiles with dependent data that were used in the simulations. The moving block bootstrap and extended tapered block bootstrap methods, labeled MBB and ETBB, were proposed by \citet{Fitzenberger98} and \citet{Shao10ss} for time-series data.  They were later generalized to a smoothed and extended version in \cite{GregoryLahiriNordman18}.  We applied them unit-wise to conduct inference for panel data.  The last method was  proposed by \cite{GalvaoParkerXiao24}, and it is denoted by WEB.  

\vspace{2mm}

\noindent \textbf{MBB and ETBB}: Let $I_{1}^\ast, I_{2}^\ast, \hdots, I_{b}^\ast$ denote a uniform random sample from $\{1, 2, \hdots, T - l + 1\}$. The variables $\{ I_{j}^\ast \}_{j=1}^b$ are starting points of the blocks.  Both of these bootstrap methods use the same basic weighting scheme:
\begin{equation} \label{eq:etbb}
    \tilde{\bm{\theta}}  = \left( \tilde{\bm{\beta}}', \tilde{\bm{\alpha}}' \right) = \argmin_{ \bm{\theta} \in \bm{\Theta} } \sum_{i=1}^N \sum_{t=1}^T  \pi_{t}^\ast \rho_\tau(y_{it} - \bm{x}_{it}' \bm{\beta} - \alpha_i), 
\end{equation}
where, given block length $l$, the weight $\pi_t^\ast$ given to observation $(i,t)$ depends on a taper function $\omega_l: \{1, 2, \ldots\} \to \RR$. This function is defined by $\omega_l(s) = w((s - 1/2)) / l)$, where $w: \RR \to [0,1]$ has support $[0,1]$, is symmetric about $1/2$, and is non-decreasing in the interval $[0, 1/2]$.  The weights $\pi_t^\ast$ are
\begin{equation} \label{eq:w:etbb}
    \pi_{t}^\ast = \frac{1}{b \| \omega_l \|_1} \sum_{j=1}^b \sum_{s=1}^l \omega_l(s) I(I_j^\ast = t - s + 1)
\end{equation}
where $\| \omega_l \|_1 = \sum_{k=1}^l | \omega_l(s) |$. The block bootstraps differ only in their choice of the function $w$ used to construct the weights.  The extended tapered block bootstrap (ETBB) estimator for panel data used in simulations is obtained by using the triangular taper function $w(t) = 2t I(t \in [0, 1/2]) + 2(1 - t) I(t \in [1/2, 1])$. The moving block bootstrap (MBB) estimator uses the rectangular function $w(t) = I(t \in [0,1])$.

\vspace{2mm}

\noindent \textbf{WEB}: The bootstrap estimator proposed in \cite{GalvaoParkerXiao24} can be obtained by using weights $\pi_i^\ast$ instead of $\pi_t^\ast$ in~\eqref{eq:etbb}, where $\pi_i^\ast$ are independent and identically distributed (i.i.d.) non-negative random weights with mean and variance both equal to one. In the simulations, the weights are drawn independently from a standard exponential distribution. 

\newpage

\bibliographystyle{econometrica}
\bibliography{panel}

\end{document}